\documentclass[11pt]{article}
\usepackage[utf8]{inputenc}
\usepackage{lmodern}
\usepackage[T1]{fontenc}
\usepackage[protrusion=true,expansion]{microtype}
\usepackage[fleqn,leqno]{amsmath}
\usepackage{amsfonts}
\usepackage{amssymb}
\usepackage{amsthm}
\usepackage{mathtools}
\usepackage{enumerate}
\usepackage{graphicx}
\usepackage{booktabs}
\usepackage{xcolor}
\usepackage{hyperref}
\hypersetup{
  colorlinks=true,
  linkcolor={blue!50!black},
  citecolor={blue!50!black},
  urlcolor={blue!50!black},
}

\usepackage{geometry}
\usepackage[longnamesfirst]{natbib}
\setcitestyle{citesep={,}}

\theoremstyle{plain}
\newtheorem{theorem}{Theorem}
\newtheorem*{theorem*}{Theorem}
\newtheorem{lemma}{Lemma}

\newtheorem{corollary}{Corollary}

\theoremstyle{definition}
\newtheorem*{definition}{Definition}
\newtheorem{assumption}{Assumption}
\newtheorem{example}{Example}

\theoremstyle{remark}
\newtheorem*{remark}{Remark}

\DeclareMathOperator{\E}{\mathrm{E}}
\DeclareMathOperator{\1}{\mathrm{1}}
\newcommand{\abs}[1]{\left\lvert#1\right\rvert}

\newcommand{\thetaM}{\theta_{\mathrm{M}}}
\newcommand{\thetaC}{\theta_{\mathrm{C}}}
\newcommand{\thetaEC}{\theta_{\mathrm{EC}}}
\newcommand{\thetaSV}{\theta_{\mathrm{SV}}}
\newcommand{\lambar}{\bar{\lambda}}
\newcommand{\lambarminus}{\tilde{\lambda}}
\newcommand{\VM}{\mathcal{V}_{M}}
\DeclareMathOperator{\spn}{sp}
\newcommand{\cF}{c_{F}}
\newcommand{\cFp}{c_{F}'}
\newcommand{\cFLogit}{1/4}
\newcommand{\cFpLogit}{\sqrt{3}/18}

\newcommand{\NSample}{100{,}000{,}000}
\newcommand{\NSampleWords}{100~million}
\newcommand{\PctUnique}{96.69\%}
\newcommand{\PctMultiple}{3.31\%}
\newcommand{\PctThreeEq}{3.30\%}
\newcommand{\PctFiveOrMoreEq}{0.0019\%}
\newcommand{\ThmOneRecall}{1.0\%}
\newcommand{\ThetaMMin}{0}
\newcommand{\ThetaMMax}{5}
\newcommand{\ThetaCMin}{-5}
\newcommand{\ThetaCMax}{0}
\newcommand{\ThetaECMin}{-5}
\newcommand{\ThetaECMax}{0}
\newcommand{\ThetaSVMin}{0}
\newcommand{\ThetaSVMax}{5}
\newcommand{\RhoMin}{0.01}
\newcommand{\RhoMax}{0.5}
\newcommand{\LamMin}{0.001}
\newcommand{\LamMax}{5}
\newcommand{\ThmTwoRecallOneFifth}{7.7\%}
\newcommand{\ThmTwoRecallThree}{81.4\%}
\newcommand{\ThmOneFrontier}{2.2}
\newcommand{\ThmTwoFrontier}{0.43}
\newcommand{\ThmTwoSliceN}{7.8~million}
\newcommand{\ThmTwoSliceHolds}{94.6\%}
\newcommand{\ThmTwoSliceFrontier}{1.19}
\newcommand{\MultShareRelevant}{16\%}
\newcommand{\NMultiSeeds}{3{,}305{,}491}
\newcommand{\MedianDuoSpread}{0.128}
\newcommand{\NinetyDuoSpread}{0.337}
\newcommand{\MedianFOneSpread}{0.810}
\newcommand{\NinetyFOneSpread}{0.987}
\newcommand{\PctCanonicalStability}{99.49\%}
\newcommand{\PctUniqueStable}{99.999\%}
\newcommand{\MedianPSFlowSpread}{0.164}
\newcommand{\MedianFirmValueSpread}{0.238}
\newcommand{\ZeroStablePct}{0.45\%}
\newcommand{\ZeroStableMedRadius}{1.21}
\newcommand{\ZeroStableMedRho}{0.02}
\newcommand{\ZeroStableMedLambda}{0.15}
\newcommand{\ZeroStableMedComp}{3.1}
\newcommand{\MCReps}{500}
\newcommand{\MCMarkets}{200}
\newcommand{\MCHorizon}{25}
\newcommand{\MCThetaCFar}{\ensuremath{-2}}
\newcommand{\MCThetaCNear}{\ensuremath{-3}}
\newcommand{\MCThetaCMultiple}{\ensuremath{-4}}
\newcommand{\NPLRadiusStable}{0.065}
\newcommand{\NPLRadiusUnstable}{1.463}
\newcommand{\NPLRadiusNear}{0.968}
\newcommand{\CTNPLNearConverged}{99.80\%}
\newcommand{\CTNPLNearMedianIter}{119}
\newcommand{\NPLMaxStages}{2000}
\newcommand{\NPLTolExp}{-6}
\newcommand{\MCThetaCBifurcation}{\ensuremath{-3.1}}
\newcommand{\NCfSeeds}{10{,}000}
\newcommand{\CfMultiShareQuarter}{86\%}
\newcommand{\CfMultiShareFull}{53\%}
\newcommand{\CfOffsetN}{9{,}090}
\newcommand{\CfOffsetThmShare}{61\%}
\newcommand{\CfOffsetUniqueShare}{78\%}

\title{Multiplicity and Uniqueness of Equilibria in Continuous-Time Dynamic Discrete Games}
\author{Jason R. Blevins and Youngjae Jeong \\ The Ohio State University}
\date{October 24, 2026}

\begin{document}

\maketitle

\begin{abstract}
We consider multiplicity and uniqueness of equilibria in
continuous-time dynamic discrete games where players move
sequentially according to independent Poisson processes.
With multiple equilibria, the structural model is incomplete: the
reduced form and counterfactual predictions depend on how equilibria
are selected.
We provide closed-form
sufficient conditions for uniqueness of Markov perfect equilibrium
in this class of games.
The conditions are semiparametric and can be verified from model
primitives.
To measure the coverage of the conditions and characterize multiplicity more broadly,
we carry out a large-scale numerical search for equilibria of a dynamic
duopoly model of entry and exit.
Although uniqueness is the norm in our broad search, we find that
multiplicity concentrates in economically relevant regions: where the
market cannot profitably sustain both firms, where firms are very
patient, and at intermediate decision frequencies that balance
strategic interaction and commitment.
When multiplicity is present, we typically find two stable asymmetric
equilibria and one unstable symmetric equilibrium.
Monte Carlo and counterfactual simulations show this structure matters
for empirical work.
A full-solution estimator that assumes uniqueness can converge to the
unstable equilibrium, even when the data were generated by a stable
one, while a two-step estimator recovers the correct equilibrium.
\end{abstract}



\section{Introduction}

Multiple equilibria complicate both estimation of dynamic
discrete games and interpretation of counterfactuals.
One motivation for specifying such games in continuous
time---where decision times arrive according to independent
Poisson processes, as in the framework of 
\citet*{abbe-2016} (henceforth ABBE)---is that players move
asynchronously with probability one.
This serves the dual purposes of removing simultaneous moves as a
source of multiplicity and, because state transitions occur
sequentially, reducing the computational burden of calculating players'
expectations over successor states \citep{doraszelski-judd-2012}.
However, removing simultaneity does not fully eliminate
multiplicity, and so continuous-time models, like their discrete-time
counterparts, can exhibit multiple Markov perfect equilibria (MPE).
This possibility led ABBE and \cite{blevins-kim-2024} to
propose two-step estimation methods that condition on the
equilibrium played in the data.
Direct theoretical analysis of multiplicity in dynamic games is
complicated by the fact that players' continuation values depend on
rivals' continuation values in a simultaneous system of nonlinear
equations.
In this paper, we derive what are, to our knowledge, the first
closed-form sufficient conditions for uniqueness of MPE in
continuous-time dynamic discrete games.
Without strong restrictions on the structure of the game, the
analytical conditions we can obtain are conservative, so we turn to a
large-scale numerical search to determine where multiplicity occurs in
the parameter space, understand the economic sources of multiplicity,
and quantify its consequences for estimation and counterfactual
analysis.

The model is incomplete due to multiple equilibria, which
complicates structural estimation \citep{tamer-2003}.
In full-solution estimation, the researcher needs to either compute all
equilibria for each candidate parameter vector or impose an equilibrium
selection rule
(see \citealp{depaula-2013},
\citealp{aguirregabiria-collardwexler-ryan-2021}, and
\citealp{otsu-pesendorfer-2023} for surveys).
The literature typically addresses multiplicity by assuming that
data are generated from a single equilibrium and then either
using two-step or iterative estimation methods
\citep{aguirregabiria-mira-2007, bajari-benkard-levin-2007,
  pesendorfer08asymptotic} that build on conditional choice probability
methods for single agent models \citep{hotz93conditional},
or imposing the equilibrium conditions as constraints during
estimation \citep{su-judd-2012, egesdal15dynamic}.
The convergence of iterative methods is known to depend on the
stability of the mapping used by the estimator
\citep{pesendorfer-schmidt-dengler-2010, kasahara-shimotsu-2012, su-2014}.

Multiplicity also confounds counterfactual analysis: even if a unique
equilibrium is played in the data, it is unclear which equilibrium
will be played under counterfactual parameter values.
The predicted policy effects, therefore, depend on an equilibrium
selection assumption \citep{aguirregabiria-2012, jun-pinkse-2020}.
Sufficient conditions for uniqueness can mitigate these concerns.

Our conditions are semiparametric: they do not require a parametric
distribution for the choice-specific errors.
Our first result (Theorem~\ref{thm:contraction}) applies to the ABBE
class of dynamic games with any finite number of players, any finite
state space, and binary actions.
It provides a bound on the discount rate $\rho$ in terms of structural
primitives that ensures the Bellman system operator is a contraction
on a set of value functions that we show contains every equilibrium.

Our second result (Theorem~\ref{thm:symmetric-uniqueness}) establishes
uniqueness for a dynamic duopoly entry/exit model with symmetric firms
and offsetting switching costs (equal in magnitude, with opposite sign).
The equilibrium system reduces to two equations, allowing us to state a
less restrictive contraction condition.
When the competitive effect $\thetaC$ is moderate relative to the
move arrival rate $\lambda$ (with logit errors,
$|\thetaC| < 3.68\lambda$), the equilibrium is unique even for
arbitrarily patient firms (Corollary~\ref{cor:discount-free}).
The reduction and discount-free conditions exploit features of the
continuous-time model with asynchronous moves: the equilibrium system
is linear in rivals' choice probabilities, and firms lose their
ability to commit as moves become more frequent.
Both conditions are stated in terms of primitives, so applied
researchers can verify them to check for uniqueness in estimation
or counterfactual analysis.
In Section~\ref{sec:counterfactual}, we apply
Theorem~\ref{thm:symmetric-uniqueness} to a counterfactual entry
subsidy that eliminates sunk costs.

We investigate the range of parameters for which our conditions hold
and further characterize uniqueness and multiplicity by numerically
searching for equilibria in the duopoly model over \NSampleWords{}
parameter vectors.
We find that uniqueness is common even beyond the regions covered by
our theorems, but multiplicity appears in economically relevant cases
such as when firms are patient, competition is strong, or entry costs
are high relative to scrap values.
The multiplicity rate reaches \MultShareRelevant{} among cases where
duopoly is unprofitable, firms discount at standard rates, and moves
are not too infrequent.
It is more likely when there is a balance between strategic
interaction (sufficiently frequent moves) and commitment (sufficiently
infrequent moves).

Where multiplicity occurs, we investigate its consequences in terms of
economic outcomes, estimation, and counterfactuals.
We typically find two stable asymmetric mirror-image equilibria and
an unstable symmetric equilibrium.
The long-run duopoly share varies by \MedianDuoSpread{} across these
three equilibria in the median case.
In a Monte Carlo experiment, a full-solution maximum likelihood
estimator that assumes uniqueness can converge to an unstable
equilibrium even when the data were generated by a stable one.
In contrast, two-step estimators initialized with first-stage
estimates of choice probabilities select the equilibrium that
generated the data, and their iterated versions do so when the
iteration is contractive.
Finally, equilibrium selection can have a large effect on
counterfactual predictions involving firm identity, as we
demonstrate through an entry subsidy.

Although we focus on dynamic games, equilibrium multiplicity has
been studied more thoroughly in static games.
Beginning with the foundational work of
\cite{bresnahan91empirical} and \cite{berry-1992},
the literature on empirical entry games has focused primarily on
equilibrium multiplicity in static models
(e.g., \citealp{tamer-2003}, \citealp{jia-2008},
\citealp{ciliberto-tamer-2009}, and \citealp{sweeting-2009}).
More recently, \cite{espin-sanchez-parra-wang-2023} study static entry
games with private information and, based on a measure of competitive
strength that orders firms, provide sufficient conditions for a unique
equilibrium.

In dynamic games, a number of computational papers have documented
the existence of multiple equilibria in discrete-time models.
\cite{besanko10learning} propose a model of learning-by-doing with
organizational forgetting and find that multiplicity is concentrated
where the inflow of know-how from learning balances the outflow from
forgetting.
\cite{iskhakov16recursive} use a recursive lexicographic search
algorithm to find all MPE of directional dynamic games---games in
which some component of the state evolves irreversibly---and show that
discrete-time Bertrand models can support hundreds of millions of
equilibria.

Theoretical work has characterized equilibria in related models,
including discrete-time finite discounted stochastic games.
\cite{haller-lagunoff-2000} prove that such games generically
have a finite number of MPE, and
\cite{herings-peeters-2004} refine this to an odd number.
\cite{doraszelski-escobar-2010} show that such equilibria are
generically regular.
\cite{doraszelski-satterthwaite-2010} introduce privately observed,
randomly drawn entry costs and scrap values in the
\cite{ericson-pakes-1995} framework to ensure existence of
an MPE in cutoff entry/exit strategies.

In certain classes of discrete-time dynamic models, sufficient
conditions for uniqueness are known.
In a dynamic duopoly model of price competition with
learning-by-doing, \cite{cabral-riordan-1994} prove that the MPE is
unique, but show that adding avoidable fixed costs and irreversible
exit can result in multiplicity.
\cite{besanko10learning} extend this result to a model
with organizational forgetting, proving uniqueness in the two limiting
cases in which forgetting is either absent or certain.
In an oligopoly entry/exit model in which firms make survival decisions
sequentially (in order of seniority), \cite{abbring-campbell-2010}
establish existence of an essentially unique symmetric MPE within
the class of ``last-in first-out'' strategies.
\cite{abbring-campbell-tilly-yang-2018} instead assume simultaneous
survival decisions and obtain an essentially unique symmetric MPE
without last-in first-out restrictions, under the assumption of
homogeneous firms.
Since these results rely on special structures or equilibrium
refinements, they are not directly applicable to the general
continuous-time models considered here.

The remainder of this paper is organized as follows.
Section~\ref{sec:model} introduces the class of continuous-time games
we consider and our leading example: a duopoly entry and exit model.
Section~\ref{sec:uniqueness} establishes our main theoretical results,
providing sufficient conditions for equilibrium uniqueness.
Section~\ref{sec:numerical} develops our computational approach and
presents numerical results for the structural parameters and economic
quantities that determine regions of multiplicity and uniqueness.%
\footnote{Complete replication code is available at
  \url{https://github.com/jrblevin/ctunique-replication}.}
Section~\ref{sec:estimation} discusses the implications of
equilibrium multiplicity for structural estimation and counterfactual
analysis.
Section~\ref{sec:conclusion} concludes.

\section{Continuous-Time Dynamic Discrete Games}
\label{sec:model}

In this section we briefly review the class of continuous-time dynamic
discrete games introduced by ABBE and present the two-player
entry/exit model that will serve as our leading example throughout
the paper.

\subsection{General Framework}
\label{sec:framework}

There are $N$ players in the model, indexed by $i = 1, \dots, N$.
Each player makes discrete choices over an infinite time horizon at
stochastically determined decision epochs that arrive according to
independent Poisson processes with rates $\lambda_i$.
Therefore, the move times of players are almost surely sequential
rather than simultaneous.
There is a commonly observed state vector $x$, which is a member of
the finite state space $\mathcal{X} \equiv \{x_1, \dots, x_K\}$.
State transitions occur endogenously as a result of the actions
of players.

At a decision epoch for player $i$ in state $x_k$, the player observes
a vector of private choice-specific errors $\varepsilon_{ik}(t)$ and
chooses a discrete action $j$ from the choice set $\mathcal{J}$.
In this paper, we limit consideration to binary-choice models with
$\mathcal{J} = \{ 0, 1 \}$.
Action $j = 0$ represents continuation, with no corresponding state change,
while $j = 1$ causes the state to jump immediately and deterministically
from $x_k$ to $x_{\ell(i,k)}$,
where $\ell(i, \cdot) \colon \{1, \dots, K\} \to \{1, \dots, K\}$
denotes player $i$'s switching state map.
We encode this map as a deterministic switching matrix $S_i$ with
$[S_i]_{k, \ell(i,k)} = 1$ for each row $k$ and zeros elsewhere.

While the model remains in some state $x_k$ between state transitions,
players continuously accrue flow payoffs $u_{ik}$.
At a decision epoch $t$ where player $i$ chooses action $j$ in state $x_k$,
that player receives a lump-sum payoff
\begin{equation*}
\psi_{ijk} + \varepsilon_{ijk}(t),
\end{equation*}
where $\psi_{ijk}$ is a deterministic, instantaneous choice-specific
payoff and $\varepsilon_{ijk}(t)$ is the corresponding choice-specific
component of the error vector.
We normalize the payoff for the continuation action $j=0$ by setting
$\psi_{i0k} = 0$ for all $i$ and $k$ (Assumption~\ref{assp:dn} below).
Therefore, with binary actions we simply drop the $j$ subscript and write
$\psi_{ik} \equiv \psi_{i1k}$ for the switching cost and
$\psi_i = (\psi_{i1}, \dots, \psi_{iK})^\top$ for the vector of
switching costs across states.

The errors are i.i.d.\ across epochs (Assumption~\ref{assp:pi} below)
and do not directly enter other players' continuation values, so it is
irrelevant whether rivals observe past error vectors.
We assume each player observes only its own errors (i.e., 
the errors are private information).
Since the model is stationary, we omit the time argument and
write $\varepsilon_{ijk}$ and $\varepsilon_{ik}$ to denote a
choice-specific error component and vector, respectively,
at a generic decision epoch.

We now formally state the main assumptions of the model, closely following ABBE.

\begin{assumption}[Discrete States]
\label{assp:dx}
The state space is finite: $K \equiv |\mathcal{X}| < \infty$.
\end{assumption}

\begin{assumption}[Discount Rates]
\label{assp:rho}
Players use a common discount rate $\rho \in (0,\infty)$.
\end{assumption}

\begin{assumption}[Move Arrival Rates]
\label{assp:rates}
Move arrival times follow independent Poisson processes with rate parameters
$\lambda_i$ satisfying $0 < \lambda_i < \infty$ for $i = 1, \dots, N$.
\end{assumption}

\begin{assumption}[Bounded Payoffs]
\label{assp:u}
Flow payoffs and instantaneous payoffs are finite:
$|u_{ik}| < \infty$ and $|\psi_{ijk}| < \infty$
for $i = 1, \dots, N$, $j = 0, 1$, and $x_k \in \mathcal{X}$.
\end{assumption}

\begin{assumption}[Costless Continuation]
\label{assp:dn}
For $i = 1, \dots, N$ and for all $x_k \in \mathcal{X}$, $\psi_{i0k} = 0$.
\end{assumption}

\begin{assumption}[Distribution of Private Information]
\label{assp:pi}
The error vectors
$\varepsilon_{ik}(t) \equiv (\varepsilon_{i0k}(t), \varepsilon_{i1k}(t))$
are i.i.d.\ across players $i$, states $x_k$, and decision epochs $t$.
The distribution is absolutely continuous with respect to
Lebesgue measure, with finite first moments and support
equal to $\mathbb{R}^{|\mathcal{J}|}$.
\end{assumption}

Under Assumption~\ref{assp:pi}, the error difference
$\varepsilon_{i0k} - \varepsilon_{i1k}$ is absolutely continuously
distributed with support $\mathbb{R}$, and we let $F$ and $f$ denote
its CDF and density.

\begin{assumption}[Bounded Error-Difference Density]
\label{assp:bounded-density}
The density $f$ is bounded: $\cF \equiv \sup_z f(z) < \infty$.
\end{assumption}

\begin{assumption}[Exchangeable Errors with Lipschitz Density]
\label{assp:exchangeable}
The pair $(\varepsilon_{i0k}, \varepsilon_{i1k})$ is exchangeable
and the density $f$ is Lipschitz continuous with constant $\cFp$.
\end{assumption}

\begin{assumption}[Type I Extreme Value Errors]
\label{assp:ev1}
For each $i$ and $k$, the choice-specific errors $\varepsilon_{i0k}$
and $\varepsilon_{i1k}$ are independent and each follows the type I
extreme value distribution with mean $0$ and scale parameter $1$.
\end{assumption}

\begin{remark}
Assumption~\ref{assp:ev1} implies Assumptions~\ref{assp:bounded-density}
and~\ref{assp:exchangeable} with $\cF = \cFLogit{}$ and
$\cFp = \cFpLogit{}$ (Lemma~\ref{lem:logit-lipschitz} in
Appendix~\ref{sec:proofs}).
Although this assumption is standard in applied work and yields
closed-form choice probabilities and surplus functions, our uniqueness
results in Section~\ref{sec:uniqueness} do not require it.
We state it for use in examples and in the numerical work of
Section~\ref{sec:numerical}.
\end{remark}

This framework is a special case of the ABBE model, with no exogenous
state changes and binary actions.
The first restriction is for simplicity only.
Our main uniqueness condition holds if exogenous state changes are
allowed (see the remarks following Theorem~\ref{thm:contraction}).
Generalizing to more than two choices is a more substantive extension
that we leave for future work.

\subsection{Value Functions and Equilibrium}
\label{sec:value-functions}

A \emph{stationary Markov strategy} for player $i$ is a function
$\delta_i \colon \mathcal{X} \times \mathbb{R}^{|\mathcal{J}|} \to \mathcal{J}$
that assigns an action to each combination of the state $x_k$ and
realized error vector $\varepsilon_{ik}$.
A strategy $\delta_i$ induces conditional choice probabilities (CCPs)
\begin{equation}
\label{eq:ccp-from-policy}
\sigma_{ijk} \equiv \Pr\left[\delta_i(x_k, \varepsilon_{ik}) = j \mid x_k \right],
\end{equation}
where the probability is taken over the distribution of $\varepsilon_{ik}$
in Assumption~\ref{assp:pi}.
Since actions are binary, we track only the switching probabilities
$\sigma_{ik} \equiv \sigma_{i1k}$, with
$1 - \sigma_{ik} = \sigma_{i0k}$ being the probability of continuation.
We denote profiles of CCPs and strategies by
$\sigma = (\sigma_1, \dots, \sigma_N)$ and
$\delta = (\delta_1, \dots, \delta_N)$,
respectively, and write $\sigma_{-i}$ for the collection of CCPs of
player $i$'s rivals.

Let $V_{ik}$ denote the value function for
player $i$ in state $x_k$ given the rivals' switching
probabilities $\sigma_{-i}$ (we suppress the dependence of $V_{ik}$
on $\sigma_{-i}$, which is implicit):
\begin{multline}
\label{eq:bellman}
V_{ik} = \frac{1}{\rho + \eta}
\left[
    u_{ik}
    + \sum_{m \neq i} \lambda_m \lbrace \sigma_{mk} V_{i,\ell(m,k)} + (1-\sigma_{mk}) V_{ik} \rbrace
    \right. \\ \left.
    + \lambda_i \E_{\varepsilon} \max \{\psi_{ik} + \varepsilon_{i1k} + V_{i,\ell(i,k)}, \varepsilon_{i0k} + V_{ik} \}
\right],
\end{multline}
where $\eta \equiv \sum_{m=1}^{N} \lambda_m$ is the total move
arrival rate, so that values are discounted at the effective
rate $\rho + \eta$.
The first term inside the square brackets, $u_{ik}$, is the flow
payoff in state $x_k$.
The middle terms reflect the possibility that rival $m$ receives a
decision opportunity (at rate $\lambda_m$) and either switches (with
probability $\sigma_{mk}$), moving the state to $x_{\ell(m,k)}$, or
continues (with probability $1 - \sigma_{mk}$), leaving the state
unchanged.
These middle terms are linear in rival players' choice probabilities
since at most one player moves per decision epoch.
In a simultaneous-move discrete-time model, however, the continuation
value of player $i$ is an expectation over the joint actions of all
rivals, so rival choice probabilities enter inside the expected maximum
over player $i$'s own actions.
The final term is player $i$'s own expected option value when it
receives a decision opportunity, where $\psi_{ik}$ is the switching
cost and the expectation $\E_{\varepsilon}$ is taken over the
choice-specific errors.

Let $E_{ik}(V_i)$ denote McFadden's social surplus function
\citep{mcfadden-1978-gev, mcfadden81econometric}, which gives the
expected maximum utility for player $i$ in state $x_k$:
\begin{align*}
    E_{ik}(V_i) &\equiv \E_{\varepsilon} \max \{\psi_{ik} + \varepsilon_{i1k} + V_{i,\ell(i,k)}, \varepsilon_{i0k} + V_{ik} \} \\
    &= L(\psi_{ik} + V_{i,\ell(i,k)},\, V_{ik})\\
    &= V_{ik} + e(\psi_{ik} + V_{i,\ell(i,k)} - V_{ik}),
\end{align*}
where
\begin{equation*}
  L(z_1, z_2) \equiv \E_{\varepsilon} \max\{z_1 + \varepsilon_{i1k},\, z_2 + \varepsilon_{i0k}\}
  \quad\text{and}\quad
  e(z) \equiv L(z, 0).
\end{equation*}
Here $L(z_1, z_2)$ is the surplus from a binary choice with
deterministic payoff components $z_1$ and $z_2$, and $e(z)$
is shorthand we use throughout the analysis that follows.
The final equality for $E_{ik}$ holds because adding a constant to
both payoffs shifts the surplus by that amount.
Both functions $L$ and $e$ are finite under Assumption~\ref{assp:pi}
and do not depend on $(i, k)$ since the errors are identically
distributed across players and states.
Collecting these values into a vector $E_i(V_i)$, with
$V_i = (V_{i1}, \dots, V_{iK})^\top$, we can write
$E_i(V_i) = V_i + e(\psi_i + S_i V_i - V_i)$, where $e$ applies
componentwise.

Let $b_i \colon \mathcal{X} \times \mathbb{R}^{|\mathcal{J}|} \times
[0,1]^{(N-1)K} \to \mathcal{J}$ denote the \emph{best-response function}
for player $i$ in state $x_k$ upon observing the error vector
$\varepsilon_{ik}$ given the rivals' CCPs $\sigma_{-i}$.
Player $i$ chooses to switch ($j=1$) if and only if the choice-specific
value of switching is at least that of continuing:
\begin{equation}
\label{eq:best-response}
b_i(x_k, \varepsilon_{ik}, \sigma_{-i})
= \1\bigl\{
    \psi_{ik} + \varepsilon_{i1k} + V_{i,\ell(i,k)}
    \geq \varepsilon_{i0k} + V_{ik}
\bigr\},
\end{equation}
where $V_{ik}$ and $V_{i,\ell(i,k)}$ are the values induced by $\sigma_{-i}$
through~\eqref{eq:bellman}.

\begin{definition}
\label{def:mpe}
A \emph{Markov perfect equilibrium} (MPE) is a profile of stationary
Markov strategies $\delta^* = (\delta_1^*, \dots, \delta_N^*)$ such that, for
each player $i$ and each $(x_k, \varepsilon_{ik})$,
\begin{equation*}
\delta_i^*(x_k, \varepsilon_{ik}) = b_i(x_k, \varepsilon_{ik}, \sigma_{-i}^*),
\end{equation*}
where $\sigma^* = (\sigma_1^*, \dots, \sigma_N^*)$ is the CCP profile
derived from $\delta^*$ via~\eqref{eq:ccp-from-policy}.
\end{definition}

An MPE requires that each player's strategy is a best response
to the rivals' CCPs and that the CCPs are themselves consistent with
the strategies.
Because player $i$'s payoffs depend on the rivals only through the
distribution of the rivals' actions in each state, the strategies
$\delta^*$ enter the best-response problem only through their induced
CCPs $\sigma^*$.
The equilibrium can therefore be equivalently characterized in
distributional-strategy form \citep{milgrom-weber-1985}.
Following \cite{aguirregabiria-mira-2007} and ABBE,
we adopt this characterization and treat MPE as a
fixed point of the equilibrium condition in CCPs $\sigma$.
Applying~\eqref{eq:best-response}, the equilibrium switching
probability is
\begin{equation}
\label{eq:ccp}
\sigma_{ik}
= \Pr\bigl[
    \psi_{ik} + \varepsilon_{i1k} + V_{i,\ell(i,k)}
    \geq \varepsilon_{i0k} + V_{ik}
\bigr]
= F(\psi_{ik} + V_{i,\ell(i,k)} - V_{ik}),
\end{equation}
where $F$ is the CDF of the error difference
$\varepsilon_{i0k} - \varepsilon_{i1k}$.

\begin{example}[Logit]
\label{ex:logit}
Under Assumption~\ref{assp:ev1}, the error difference
$\varepsilon_{i0k} - \varepsilon_{i1k}$ follows the standard
logistic distribution, so $F = \Lambda$, where
$\Lambda(z) \equiv \exp(z)/(1 + \exp(z)) = 1/(1 + \exp(-z))$ is
the logistic CDF, and the switching probabilities~\eqref{eq:ccp}
take the familiar logit form.
The surplus function has the log-sum-exp closed form
$E_i(V_i) = \log[\exp(\psi_i + S_i V_i) + \exp(V_i)]$, with the
exponential and logarithm applied componentwise.
\end{example}

The values $V_{ik}$ and $V_{i,\ell(i,k)}$ in~\eqref{eq:ccp} are induced
by the rivals' CCPs $\sigma_{-i}$. We stack the switching
probabilities~\eqref{eq:ccp} for each player $i$ across states
$x_k$.
A fixed point $\sigma^* = (\sigma_1^*, \dots, \sigma_N^*)$ of this
system characterizes an MPE.
Under Assumptions~\ref{assp:dx}--\ref{assp:pi}, at least one such
fixed point exists, by Theorem~1 of \cite{ctgames}.

The equilibrium conditions can be expressed in terms of either CCPs
or value functions.
Given an equilibrium CCP profile $\sigma$, each player's value
function is uniquely determined as the solution of a linear system
with a nonsingular coefficient matrix, by Theorem~2 of
\cite{ctgames}, which generalizes Proposition~6 of ABBE.
Conversely, given the value functions, the CCPs are given directly
by~\eqref{eq:ccp}.
The value functions are therefore sufficient to determine the
CCPs, and the two fixed-point problems have equivalent solutions:
$\sigma^*$ is an equilibrium CCP profile if and only
if the associated value functions solve the Bellman
system obtained by substituting~\eqref{eq:ccp}
into~\eqref{eq:bellman}.
In the remainder of the paper we therefore characterize equilibria
as solutions of this value-function system.

Substituting the explicit forms of the choice probabilities and social
surplus functions, we can express the equilibrium system of Bellman
equations entirely in terms of value functions.
Player $m$'s CCPs induce the $K \times K$ transition probability matrix
\begin{equation}
\label{eq:Sigma-general}
\Sigma_m(V_m) \equiv I + D_m(V_m)(S_m - I),
\end{equation}
where $D_m(V_m) \equiv \operatorname{diag}(\sigma_{m1}, \dots, \sigma_{mK})$
holds the switching probabilities~\eqref{eq:ccp}.
Row $k$ of $\Sigma_m(V_m)$ places probability $\sigma_{mk}$ on the
switching destination $\ell(m,k)$ and probability $1 - \sigma_{mk}$ on
$k$, so each $\Sigma_m(V_m)$ is a stochastic matrix.
The Bellman equation~\eqref{eq:bellman} for player $i$ can then be
stated simply in vectorized form as
\begin{equation}
\label{eq:bellman-general}
V_i = \frac{1}{\rho + \eta} \Bigl[ u_i + \sum_{m \neq i} \lambda_m \Sigma_m(V_m) V_i + \lambda_i E_i(V_i) \Bigr].
\end{equation}
Defining $V \equiv (V_1^\top, \dots, V_N^\top)^\top \in \mathbb{R}^{NK}$,
the equilibrium system can be written compactly as
\begin{equation}
\label{eq:bellman-system}
V = T\left( V \right),
\end{equation}
where $T \colon \mathbb{R}^{NK} \to \mathbb{R}^{NK}$ is the
\emph{Bellman system operator} defined by
stacking~\eqref{eq:bellman-general} over players $i = 1, \dots, N$.

\subsection{A Dynamic Duopoly Model of Entry and Exit}
\label{sec:duopoly}

Our leading example, to which our second uniqueness condition and our
numerical work apply, is a two-player instance of this framework in
an entry/exit setting.
When active, the firms operate in the same product market, providing
the same good or service.

The state of the market is represented by a state vector
$x \in \mathcal{X} \equiv \{ (0,0), (0,1), (1,0), (1,1) \}$, where the first
component represents the entry status of firm~1---with 0 denoting inactive
and 1 denoting active---and the second component represents
the entry status of firm~2.
It will be convenient to index the states by integers as $x_k$
for $k = 1, 2, 3, 4$ where
$x_1 = (0,0)$, $x_2 = (0,1)$, $x_3 = (1,0)$, and $x_4 = (1,1)$.

For each firm, the switching action ($j = 1$) toggles its own
activity status: an inactive firm enters the market and an active
firm exits the market.
The continuation action ($j = 0$) leaves the firm's status---and
hence the state---unchanged.
The resulting transitions are shown in Table~\ref{tab:transitions}.

\begin{table}[t!]
\centering
\begin{tabular}{cccccccc}
\hline
\multicolumn{2}{c}{Initial State} && \multicolumn{2}{c}{Decision} && \multicolumn{2}{c}{New State} \\
$k$ & $x_k$ && Firm $i$ & Action $j$ &&  $\ell(i,k)$ & $x_{\ell(i,k)}$ \\
\hline
1 & $(0,0)$ && 1 & 1 && 3 & $(1,0)$ \\
2 & $(0,1)$ && 1 & 1 && 4 & $(1,1)$ \\
3 & $(1,0)$ && 1 & 1 && 1 & $(0,0)$ \\
4 & $(1,1)$ && 1 & 1 && 2 & $(0,1)$ \\
\hline
1 & $(0,0)$ && 2 & 1 && 2 & $(0,1)$ \\
2 & $(0,1)$ && 2 & 1 && 1 & $(0,0)$ \\
3 & $(1,0)$ && 2 & 1 && 4 & $(1,1)$ \\
4 & $(1,1)$ && 2 & 1 && 3 & $(1,0)$ \\
\hline
\end{tabular}
\caption{State transitions for switching actions ($j=1$) for each player $i=1,2$. Continuation actions ($j=0$) result in no state change.}
\label{tab:transitions}
\end{table}

In terms of the switching matrices defined in
Section~\ref{sec:framework}, the transitions in
Table~\ref{tab:transitions} correspond to
\begin{equation*}
  S_1 = \begin{bmatrix}
    0 & 0 & 1 & 0 \\
    0 & 0 & 0 & 1 \\
    1 & 0 & 0 & 0 \\
    0 & 1 & 0 & 0 \\
  \end{bmatrix}
  \quad
  \text{and}
  \quad
  S_2 = \begin{bmatrix}
    0 & 1 & 0 & 0 \\
    1 & 0 & 0 & 0 \\
    0 & 0 & 0 & 1 \\
    0 & 0 & 1 & 0 \\
  \end{bmatrix}.
\end{equation*}

The flow payoffs for the two firms can be
represented by vectors as follows:
\begin{equation}
    u_{1} = \begin{pmatrix}
        0 \\
        0 \\
        \thetaM \\
        \thetaM + \thetaC \\
    \end{pmatrix}
    \text{ and }
    u_{2} = \begin{pmatrix}
        0 \\
        \thetaM \\
        0 \\
        \thetaM + \thetaC \\
    \end{pmatrix}
    \text{ for }
    \begin{array}{c}
        x_1 = (0,0) \\
        x_2 = (0,1) \\
        x_3 = (1,0) \\
        x_4 = (1,1) \\
    \end{array}.
\end{equation}

We assume $\thetaM > 0$ and $\thetaC < 0$ in the duopoly model:
$\thetaM$ is the flow profit from operating in the market as a
monopolist, and $\thetaC$ is the change in flow profit from operating
in the market as a duopolist after the entry of the rival firm.

The switching-cost vectors are
\begin{equation}
    \label{eq:switching_costs}
    \psi_{1} = \begin{pmatrix}
        \thetaEC \\
        \thetaEC \\
        \thetaSV \\
        \thetaSV \\
    \end{pmatrix}
    \text{ and }
    \psi_{2} = \begin{pmatrix}
        \thetaEC \\
        \thetaSV \\
        \thetaEC \\
        \thetaSV \\
    \end{pmatrix}
    \text{ for }
    \begin{array}{c}
        x_1 = (0,0) \\
        x_2 = (0,1) \\
        x_3 = (1,0) \\
        x_4 = (1,1) \\
    \end{array}.
\end{equation}
Here $\thetaEC$ denotes the cost of entry (switching from 0 to 1) and
$\thetaSV$ denotes the scrap value received when the firm exits
(switching from 1 to 0).

For the duopoly entry/exit model, $N = 2$ and $K = 4$, so
\eqref{eq:bellman-system} is the full eight-equation equilibrium
system.
For firm~1, equation~\eqref{eq:bellman-general} reads
\begin{equation}
\label{eq:bellman-firm1}
V_1 = \frac{1}{\rho + \eta} \left[ u_1 + \lambda_2 \Sigma_2(V_2) V_1 + \lambda_1 E_1(V_1) \right],
\end{equation}
where the $4 \times 4$ transition probability
matrix determined by firm~2's switching probabilities is
\begin{equation}
\label{eq:Sigma_2}
\Sigma_2(V_2) = \left[\begin{smallmatrix}
1-F(\thetaEC + V_{22} - V_{21}) & F(\thetaEC + V_{22} - V_{21}) & 0 & 0 \\
F(\thetaSV + V_{21} - V_{22}) & 1-F(\thetaSV + V_{21} - V_{22}) & 0 & 0 \\
0 & 0 & 1-F(\thetaEC + V_{24} - V_{23}) & F(\thetaEC + V_{24} - V_{23}) \\
0 & 0 & F(\thetaSV + V_{23} - V_{24}) & 1-F(\thetaSV + V_{23} - V_{24})
\end{smallmatrix}\right]
\end{equation}
and firm~1's social surplus vector
$E_1(V_1) = V_1 + e(\psi_1 + S_1V_1 - V_1)$ is, explicitly,
\begin{equation*}
E_1(V_1) = \begin{bmatrix}
V_{11} + e(\thetaEC + V_{13} - V_{11}) \\
V_{12} + e(\thetaEC + V_{14} - V_{12}) \\
V_{13} + e(\thetaSV + V_{11} - V_{13}) \\
V_{14} + e(\thetaSV + V_{12} - V_{14})
\end{bmatrix}.
\end{equation*}

Similarly for firm~2,
\begin{equation}
\label{eq:bellman-firm2}
V_2 = \frac{1}{\rho + \eta} \left[ u_2 + \lambda_1 \Sigma_1(V_1) V_2 + \lambda_2 E_2(V_2) \right],
\end{equation}
where
\begin{equation}
\label{eq:Sigma_1}
\Sigma_1(V_1) = \left[\begin{smallmatrix}
1-F(\thetaEC + V_{13} - V_{11}) & 0 & F(\thetaEC + V_{13} - V_{11}) & 0 \\
0 & 1-F(\thetaEC + V_{14} - V_{12}) & 0 & F(\thetaEC + V_{14} - V_{12}) \\
F(\thetaSV + V_{11} - V_{13}) & 0 & 1-F(\thetaSV + V_{11} - V_{13}) & 0 \\
0 & F(\thetaSV + V_{12} - V_{14}) & 0 & 1-F(\thetaSV + V_{12} - V_{14})
\end{smallmatrix}\right]
\end{equation}
and
$E_2(V_2) = V_2 + e(\psi_2 + S_2V_2 - V_2)$.

\section{Conditions for Equilibrium Uniqueness}
\label{sec:uniqueness}

Having characterized the equilibrium system, and with existence
of an equilibrium guaranteed, we now establish conditions under
which the equilibrium is unique.
Equilibria are solutions of the fixed-point problem $V = T(V)$,
where $T$ is the Bellman system operator defined
in~\eqref{eq:bellman-system}.
Our approach proceeds in two steps.
First, we derive an a priori bound on equilibrium values: the
span---the difference between the largest and smallest
components---of \emph{any} equilibrium value function is bounded by a
constant $M$ determined by the model primitives, so every
equilibrium lies in the sublevel set $\VM$ of value functions whose
spans are at most $M$.
This set is \emph{invariant} under $T$ in the sense that $T$ maps
it into itself.
Second, we provide sufficient conditions on the model parameters
under which $T$ is a contraction on this set---that is, there
exists $\gamma < 1$ such that
$\|T(V) - T(V')\|_\infty \leq \gamma \|V - V'\|_\infty$
for all $V$, $V'$ in $\VM$.
By the Banach fixed-point theorem, $T$ then has a unique fixed
point in the set, and since every equilibrium lies in $\VM$, the
equilibrium is globally unique.
Establishing contractivity requires bounding the Lipschitz constant
of $T$, which in turn requires controlling the derivatives of the
choice probabilities and the social surplus functions.
The $\Sigma_m(V_m)$ matrices, which depend on the rivals'
values, couple the players' Bellman equations and make the
contraction property more difficult to establish than in
single-agent dynamic programming problems.

Notably, our results are semiparametric with respect to the error
distribution and do not rely on the logit structure of
Assumption~\ref{assp:ev1}.
Theorem~\ref{thm:contraction} requires
Assumption~\ref{assp:bounded-density} (bounded density),
and Theorem~\ref{thm:symmetric-uniqueness} adds
Assumption~\ref{assp:exchangeable} (exchangeability and Lipschitz
density).

Let $\lambar \equiv \max_i \lambda_i$ denote the maximal move arrival
rate, let $\lambda_{-i} \equiv \sum_{m \neq i} \lambda_m$ denote the
total arrival rate of player $i$'s rivals, and let
$\lambarminus \equiv \max_i \lambda_{-i}$.
In the duopoly model, $\lambda_{-i}$ is the rate of the single rival, and
$\lambarminus = \max\{\lambda_1, \lambda_2\} = \lambar$.
For any vector $z \in \mathbb{R}^n$, define the
\emph{span seminorm} \citep[see][pp.~195--196]{puterman05markov}
as\footnote{This is a seminorm:
  $\spn(z) \geq 0$,
  $\spn(z + z') \leq \spn(z) + \spn(z')$, and
  $\spn(\alpha z) = |\alpha|\spn(z)$.
  Its kernel consists of the constant vectors:
  $\spn(z) = 0$ if and only if $z$ is constant.}
\begin{equation*}
\spn(z) \equiv \max_k z_k - \min_k z_k.
\end{equation*}
Using this notation, define
\begin{equation}
\label{eq:U-def}
U \equiv \max_i \spn(u_i),
\quad
\Psi \equiv \max_i \spn(\psi_i),
\quad \text{and} \quad
R(V) \equiv \max_i \spn(V_i),
\end{equation}
the maximum range of the flow payoffs, switching costs, and
value functions over players, respectively.
For the duopoly entry/exit model, $U$ and $\Psi$ reduce to
\begin{equation*}
U = \max\{\abs{\thetaM}, \abs{\thetaM + \thetaC}, \abs{\thetaC}\}
\quad \text{and} \quad
\Psi = \abs{\thetaEC - \thetaSV}.
\end{equation*}

The following lemma establishes the preliminary results that form
the basis of our uniqueness conditions.
Part~(a) bounds the span of $T(V)$ in terms of the span of $V$ for
arbitrary $V \in \mathbb{R}^{NK}$ (not necessarily an equilibrium).
The bound is an affine contraction in the span with modulus
$\eta/(\rho + \eta)$, and the constant $M$ defined in the lemma is its
steady state.
Part~(b) establishes that the sublevel set
$\VM = \{V : R(V) \leq M\}$ is invariant under $T$, which will allow
us to apply the Banach fixed-point theorem on $\VM$.
Part~(c) shows that every equilibrium must lie in $\VM$.

\begin{lemma}[Operator Bound and Invariant Set]
\label{lem:value-bounds}
Suppose Assumptions~\ref{assp:dx}--\ref{assp:pi} hold.
Let $M \equiv (U + \lambar \Psi)/\rho$ and
$\VM \equiv \{V \in \mathbb{R}^{NK} : R(V) \leq M\}$.
\begin{enumerate}[(a)]
\item \label{lem:vb-operator}
  For any $V \in \mathbb{R}^{NK}$,
  \begin{equation}
  \label{eq:operator-bound}
  R(T(V)) \leq \frac{\eta}{\rho + \eta}\, R(V)
    + \frac{U + \lambar \Psi}{\rho + \eta}.
  \end{equation}
\item \label{lem:vb-invariant}
  $T\left( \VM \right) \subseteq \VM$.
\item \label{lem:vb-equilibrium}
  If $V^*$ is an equilibrium, so that $T(V^*) = V^*$,
  then $V^* \in \VM$.
\end{enumerate}
\end{lemma}

The proof of Lemma~\ref{lem:value-bounds}, along with all other
proofs, can be found in Appendix~\ref{sec:proofs}.

Three features of the bound $M$ are worth noting.
First, it depends on move arrival rates only through
$\lambar = \max_i \lambda_i$.
Therefore, two parameter vectors with the same total rate
$\sum_i \lambda_i$ but differing in how the rates are split among
players can yield different bounds.
Second, as players become more patient ($\rho \to 0$), the modulus
$\eta/(\rho + \eta)$ in~\eqref{eq:operator-bound} approaches one, so
the contraction in the span becomes weaker and its steady state $M$
tends to infinity.
For very patient players, small differences in flow payoffs across
states can result in large differences in the value function, making
it more difficult to guarantee uniqueness.
Third, even when the span of flow payoffs $U$ is small, variation in
switching costs across states inflates the bound through the $\Psi$
term, so small differences in flow payoffs may not yield
a tight bound on the span of the value function.

\subsection{Sufficient Condition for Uniqueness}

Since every equilibrium lies in $\VM$
(Lemma~\ref{lem:value-bounds}(\ref{lem:vb-equilibrium})), it suffices
to establish uniqueness on $\VM$.
Theorem~\ref{thm:contraction} gives a sufficient condition for
$T$ to be a contraction on $\VM$.
The condition requires the discount rate to be sufficiently large
relative to a quantity depending on the span of flow payoffs $U$,
switching costs $\Psi$, and the move arrival rates.

\begin{theorem}[Contraction-Based Uniqueness]
\label{thm:contraction}
Suppose Assumptions~\ref{assp:dx}--\ref{assp:bounded-density} hold.
If
\begin{equation}
\label{eq:contraction-bound}
\rho > \sqrt{2 \cF \lambarminus (U + \lambar \Psi)},
\end{equation}
then there is a unique equilibrium.
\end{theorem}

\begin{remark}
Similar contraction arguments date back to \cite{shapley-1953},
who showed, in two-player zero-sum stochastic games, that a positive
probability of the game stopping at each stage (equivalent to
discounting when the probability is constant) makes the value
operator a contraction.
Theorem~\ref{thm:contraction} gives an analogous but stronger
requirement---discounting large enough relative to the
strategic interaction---under which the equilibrium is
globally unique in our nonzero-sum, continuous-time setting.
\cite{curtat-1996, horst-2005} introduce similar diagonal-dominance
conditions in the context of discrete-time stochastic games with
continuous states and actions, to ensure existence rather than
uniqueness of equilibrium.
\end{remark}

\begin{remark}
The factor of $2\cF$ appears in~\eqref{eq:contraction-bound}
because each rival's choice probability
$\sigma_{mk} = F(\psi_{mk} + V_{m,\ell(m,k)} - V_{mk})$ depends on
the difference of two components of $V_{m}$.
The factor $\lambarminus$ under the square root is related to
the frequency of rival moves.
The factor $U + \lambar\,\Psi$ is related to the equilibrium
value span $M$.
Flow payoffs are incurred directly, while switching costs are
scaled by the rate $\lambar$ since they are incurred at decision
epochs.
\end{remark}

\begin{remark}
The number of players enters the condition in
\eqref{eq:contraction-bound} only through $\lambarminus$.
This is another implication of asynchronous moves: each rival's
transition term appears additively in the Bellman equation for player
$i$, so the rivals' move rates are summed in the Lipschitz bound.
When players have a common arrival rate $\lambda_i = \lambda$,
we have $\lambarminus = (N-1)\lambda$ and the condition becomes
$\rho > \sqrt{2 \cF (N-1)\lambda (U + \lambda \Psi)}$.
On the other hand, in single-agent models we have $\lambarminus = 0$
and~\eqref{eq:contraction-bound} holds for all $\rho > 0$.
\end{remark}

\begin{remark}
We abstract away from exogenous state changes in Section~\ref{sec:framework}.
However, Lemma~\ref{lem:value-bounds} and Theorem~\ref{thm:contraction} still
hold if they are incorporated.
Suppose, following ABBE, that the state transitions exogenously from $x_k$ to
$x_l$ at finite rates $q_{kl}$.
For any $\bar{q} \geq \max_k \sum_{l \neq k} q_{kl}$, the total rate of events
increases to $\eta = \bar{q} + \sum_m \lambda_m$.
Let $P$ denote the corresponding transition matrix implied by uniformization
at rate $\bar{q}$.
The main change is then the addition of a term
$\bar{q} P V_i$ inside the square brackets in~\eqref{eq:bellman-general}.
Since $P$ does not depend on $V$ and stochastic matrices are nonexpansive
under both $\spn$ and $\|\cdot\|_\infty$, this term contributes $\bar{q}\,R(V)$
to the span bound of Lemma~\ref{lem:value-bounds} and
$\bar{q}\,\|V_i - V'_i\|_\infty$ to the Lipschitz bound of
Theorem~\ref{thm:contraction}, so both proofs proceed as before with the
new value of $\eta$.
\end{remark}

\subsection{Uniqueness under Offsetting Switching Costs}
\label{sec:symmetric}

The sufficient condition in the previous section bounds the
Lipschitz constant of the Bellman system operator $T$ on the
$NK$-dimensional set $\VM$.
We now return to the duopoly entry/exit model of
Section~\ref{sec:duopoly} and show that when switching costs satisfy
an offsetting condition, the eight-equation equilibrium system
can be reduced to two equations, and the resulting two-dimensional
map admits a sharper contraction condition.
The key restriction in this section is that entry costs and scrap
values are equal in magnitude.

\begin{assumption}[Offsetting Switching Costs]
\label{assp:symmetric-costs}
$\thetaEC = -\kappa$ and
$\thetaSV = \kappa$ for some $\kappa \geq 0$.
\end{assumption}

This condition is satisfied, for example, when entry requires
an investment of $\kappa$ that is fully recovered upon exit, so that
entry involves no sunk costs.
It yields the following simplified structure for the choice
probabilities, social surplus, and transition probability matrices
in the game.

\begin{lemma}[Identities with Offsetting Switching Costs]
\label{lem:symmetric-identities}
In the duopoly entry/exit model, under
Assumptions~\ref{assp:dx}--\ref{assp:exchangeable}
and~\ref{assp:symmetric-costs}:
\begin{enumerate}[(a)]
  \item \label{lem:si-probs}
  $F(\thetaEC + d) + F(\thetaSV - d) = 1$ for all $d \in \mathbb{R}$.
  \item \label{lem:si-surplus}
  $e(\thetaSV - d) - e(\thetaEC + d) = \kappa - d$ for all
  $d \in \mathbb{R}$.
  \item \label{lem:si-rows}
  For all $V \in \mathbb{R}^8$, the $2 \times 2$ blocks of
  $\Sigma_2(V_2)$ for state pairs
  $(1,2)$ and $(3,4)$ each have identical rows and therefore
  \begin{equation}
  \label{eq:sigma2-rows}
  [\Sigma_2(V_2)\, V_1]_1 = [\Sigma_2(V_2)\, V_1]_2
  \quad\text{and}\quad
  [\Sigma_2(V_2)\, V_1]_3 = [\Sigma_2(V_2)\, V_1]_4.
  \end{equation}
  Similarly, rows 1 and 3 of $\Sigma_1(V_1)$ are identical, as are
  rows 2 and 4, and therefore
  \begin{equation}
  \label{eq:sigma1-rows}
  [\Sigma_1(V_1)\, V_2]_1 = [\Sigma_1(V_1)\, V_2]_3
  \quad\text{and}\quad
  [\Sigma_1(V_1)\, V_2]_2 = [\Sigma_1(V_1)\, V_2]_4.
  \end{equation}
\end{enumerate}
\end{lemma}

Using the identities in Lemma~\ref{lem:symmetric-identities}, the
eight-dimensional equilibrium system can be reduced to a
two-dimensional problem in value differences, which makes it easier to
establish equilibrium uniqueness analytically.

\begin{lemma}[Constant Competitive Effect]
\label{lem:constant-delta}
In the duopoly entry/exit model, suppose
Assumptions~\ref{assp:dx}--\ref{assp:exchangeable}
and~\ref{assp:symmetric-costs} hold.
In any equilibrium,
\begin{equation}
\label{eq:constant-delta}
(V_{14} - V_{12}) - (V_{13} - V_{11})
= (V_{24} - V_{23}) - (V_{22} - V_{21})
= \thetaC / (\rho + \eta).
\end{equation}
\end{lemma}

For firm~1, $V_{13} - V_{11}$ represents the value gained by entering
when firm~2 is out of the market, and $V_{14} - V_{12}$ is the value
gained by entering when firm~2 is active.
Lemma~\ref{lem:constant-delta} establishes that in equilibrium the
difference between these two differences---the effect of competition
on the value of entry---is the same for both firms and equals the
competitive effect $\thetaC$ discounted at the effective rate
$\rho + \eta = \rho + \lambda_1 + \lambda_2$.
The equality in the lemma is exact due to linearity of each Bellman
equation in the rival's choice probabilities, together with
Lemma~\ref{lem:symmetric-identities}(\ref{lem:si-probs}); these terms
need not cancel under simultaneous moves.
The following theorem uses this simplification to reduce the
equilibrium mapping to two dimensions and establish a sufficient
condition for contractivity.

\begin{theorem}[Uniqueness with Offsetting Switching Costs]
\label{thm:symmetric-uniqueness}
In the duopoly entry/exit model, suppose
Assumptions~\ref{assp:dx}--\ref{assp:exchangeable}
and~\ref{assp:symmetric-costs} hold
and that $\lambda_1 = \lambda_2 = \lambda$.
If
\begin{equation}
\label{eq:symmetric-contraction}
\cF|\tilde\Delta| + (\cF^2 + \cFp)\,\tilde\Delta^2 < 1,
\quad\text{where}\quad \tilde\Delta \equiv \lambda\thetaC/[(\rho + \lambda)(\rho + 2\lambda)],
\end{equation}
then there is a unique equilibrium.
The unique equilibrium is symmetric:
$V_{11} = V_{21}$, $V_{12} = V_{23}$, $V_{13} = V_{22}$,
and $V_{14} = V_{24}$.
\end{theorem}

\begin{corollary}
\label{cor:discount-free}
In the duopoly entry/exit model, suppose
Assumptions~\ref{assp:dx}--\ref{assp:exchangeable}
and~\ref{assp:symmetric-costs} hold with
$\lambda_1 = \lambda_2 = \lambda$.  If
\begin{equation*}
\cF \frac{|\thetaC|}{2\lambda}
+ (\cF^2 + \cFp) \left(\frac{\thetaC}{2\lambda}\right)^2 < 1,
\end{equation*}
then there is a unique equilibrium for every discount rate $\rho > 0$.
\end{corollary}

In contrast to Theorem~\ref{thm:contraction}, this condition does not
involve the discount rate.
It requires that the strength of competition is moderate relative to
the move arrival rate: with logit errors, it holds whenever
$|\thetaC| < 3.68\lambda$.
In~\eqref{eq:symmetric-contraction}, $\tilde\Delta = \lambda\Delta/(\rho+\lambda)$,
where $\Delta = \thetaC/(\rho+\eta)$ is the competitive effect of
Lemma~\ref{lem:constant-delta}.
The first term can be related to the static case: for a two-player
binary game of incomplete information with interaction effect
$\tilde\Delta$, $\cF|\tilde\Delta| < 1$ implies a unique equilibrium
\citep[see][eq.~15]{aradillas-lopez-2020}.
The second term captures the dynamic effect of each firm's own
future decisions and the curvature of the error CDF $F$.
With logit errors, it tightens the condition from $|\tilde\Delta| < 4$
to $|\tilde\Delta| < 1.84$.

\begin{remark}
The left-hand side of condition~\eqref{eq:symmetric-contraction}
is nonmonotone in $\lambda$.
It is strictly increasing in $|\tilde\Delta|$, so both $|\tilde\Delta|$
and the left-hand side are maximized at
$\lambda^* = \rho/\sqrt{2}$.
The condition is easier to satisfy at extreme values of $\lambda$,
since $\tilde\Delta \to 0$ both as $\lambda \to 0$, where strategic
interaction is removed, and as $\lambda \to \infty$, where commitment
becomes more difficult.
The condition is harder to satisfy at intermediate rates
$\lambda \approx \rho/\sqrt{2}$, where strategic interaction and
commitment are both present.
\end{remark}

\section{Numerical Experiments}
\label{sec:numerical}

In this section, we investigate numerically the prevalence of
uniqueness and multiplicity in the duopoly model of entry and exit.
We evaluate the theoretical conditions, measure how far uniqueness
extends beyond them, and characterize the parameter regions where
multiplicity arises.

\subsection{Computational Approach}
\label{sec:computational}

We solve for equilibria of the dynamic duopoly model of
Section~\ref{sec:duopoly} by applying Newton's method to the system
$H(V) = V - T(V) = 0$ from multiple starting values.
We employ the analytical Jacobian, which is derived in
Appendix~\ref{sec:jacobian}.
Mapping the boundary between uniqueness and multiplicity across the
parameter space requires solving the eight-equation system hundreds
of billions of times, making the quadratic convergence speed of
Newton's method important.
Computation at this scale is possible because the continuous-time
model is computationally light: with sequential moves, the cost of
each Bellman evaluation grows linearly, rather than exponentially,
in the number of players (\citealp{doraszelski-judd-2012}; ABBE).

Although the problem of computing all Nash equilibria of static games
has a rich history \citep[see][for a survey]{mckelvey-mclennan-1996},
this is not the case for dynamic games.
Homotopy methods have been applied to trace out equilibrium correspondences
in dynamic games \citep{borkovsky10homotopy, borkovsky-doraszelski-kryukov-2012},
but all-solutions homotopies are available only for systems with special
structure, such as polynomial systems.
Our Bellman system is not polynomial, so applying homotopy methods would
essentially be a different multi-start search heuristic.

We first generate \NSample{} parameter vectors from a Halton low-discrepancy sequence over
$\thetaM \in [\ThetaMMin, \ThetaMMax]$, $\thetaC \in [\ThetaCMin, \ThetaCMax]$,
$\thetaEC \in [\ThetaECMin, \ThetaECMax]$, $\thetaSV \in [\ThetaSVMin, \ThetaSVMax]$,
$\rho \in [\RhoMin, \RhoMax]$, and $\lambda \in [\LamMin, \LamMax]$.
We focus on the case of symmetric rates $\lambda_1 = \lambda_2 = \lambda$ for simplicity
and assume that the errors follow a type I extreme value distribution,
as stated in Assumption~\ref{assp:ev1}.
With the logit scale normalized to one, payoffs are interpreted in
units of the errors.
The time unit is left unspecified, so rates and flow payoffs matter
only through ratios such as $\lambda/\rho$ and $\thetaM/\rho$;
references to standard discount rates below assume annual units.
For each parameter vector, we perform a multi-start search from
5{,}000 initial value functions.%
\footnote{In total, the search performs over 500 billion Newton solves.
The implementation uses Fortran 2008 with OpenMP parallelization.
Our results were obtained using GNU Fortran 15.1 and LAPACK 3.8.0.
The Newton solver uses a relative convergence tolerance of $10^{-13}$,
with at most 500 iterations per solve.
The search is fully deterministic: parameter vectors and starting
values are generated by Halton sequences indexed by integer seeds,
so the results are exactly reproducible.
The full search required 37.9 hours on a single 48-core cluster node
(Intel Xeon Platinum 8260).}
Six starting values are deterministic: the zero vector, the vector
of perpetuity values $u_{ik}/\rho$, and refinements of these based on
20 and 100 iterations of the Bellman system operator $T$.
Since $\lambda_1 = \lambda_2$, the game is expected to admit a symmetric
equilibrium.
To target such equilibria, we also draw 500 starting values from the
subspace of symmetric value functions.
The remaining 4{,}494 starting values are quasi-random Halton draws from
a rectangle that contains every equilibrium.%
\footnote{Following the proof of Lemma~\ref{lem:value-bounds},
we can obtain level bounds: each transition term of the Bellman system operator
lies between $\min_k V_{ik}$ and $\max_k V_{ik}$, and the logit surplus
satisfies $\max\{z, 0\} \leq \log(1 + e^z) \leq \max\{z, 0\} + \log 2$,
so any fixed point satisfies
$\rho \max_k V_{ik} \leq \max_k u_{ik}
+ \lambda_i(\max\{0, \thetaEC, \thetaSV\} + \log 2)$
and $\rho \min_k V_{ik} \geq \min_k u_{ik}$.
Every equilibrium therefore satisfies
$\min\{0, \thetaM, \thetaM + \thetaC\}/\rho \leq V_{ik}
 \leq [\max\{0, \thetaM, \thetaM + \thetaC\}
 + \lambda_i(\max\{0, \thetaEC, \thetaSV\} + \log 2)]/\rho$
for all $i, k$.}
For each asymmetric solution found, we exchange the firm values and
refine by Newton's method to recover the mirror-image counterpart.
Distinct equilibria are classified by clustering the conditional
choice probabilities in the sup norm with a tolerance of
$10^{-3}$.
We also classify each equilibrium $V^*$ as stable or unstable under
iteration of the Bellman system operator $T$: $V^*$ is stable if the
spectral radius $\rho_{\mathrm{DT}}$ of the Jacobian $DT(V^*)$ is less than
one, so that iterating $T$ from nearby starting values converges to
$V^*$.

For robustness, after completing searches from each of the 5{,}000
starting values, we check three additional conditions to determine
whether to continue with additional starting values or stop.
If an even number of equilibria (including zero) is found, no symmetric
equilibrium is found, or a unique unstable equilibrium is found,
then the search continues up to two more times.
Each of these conditions indicates a missed equilibrium: existence and
generic oddness (by analogy with the discrete-time result of
\citealp{herings-peeters-2004}) imply an odd count, a symmetric game
is expected to
have a symmetric equilibrium, and in the standard configuration we
find below, a lone unstable equilibrium suggests the stable
mirror-image pair may have been missed.
For each continuation, we add 2{,}000 Halton draws
(refined by 20 iterations of $T$ on the first retry and 100 on the
second) and 500 symmetric starting values.
Even a thorough multi-start search is not guaranteed to find all
equilibria \citep[see][]{besanko10learning}, so in the following
``unique'' means that the search described returned a single
equilibrium.

\subsection{Prevalence of Uniqueness}

\begin{table}[t]
\centering
\caption{Distribution of the number of equilibria}
\label{tab:neq-distribution}
\begin{tabular}{crr}
\toprule
Number of Equilibria & Count & Percent \\
\midrule
1 & 96{,}694{,}508 & 96.69\% \\
3 & 3{,}303{,}595 & 3.30\% \\
5 & 1{,}884 & 0.0019\% \\
7 & 12 & $<\!0.0001$\% \\
\midrule
Total & 99{,}999{,}999 & 100.00\% \\
\bottomrule
\end{tabular}
\begin{flushleft}
\small \textit{Notes:} Based on 100{,}000{,}000 parameter vectors from a Halton
sequence.
The search was incomplete in 1 case, finding no equilibrium.
This case is excluded from the total.
\end{flushleft}
\end{table}

Table~\ref{tab:neq-distribution} reports the distribution of the
number of equilibria across the sampled parameter vectors.
Overall, the equilibrium is unique in \PctUnique{} of
cases, but this rate reflects the wide sampling region.
In \PctThreeEq{} of cases, the solver found three equilibria.
This is by far the most common outcome among non-unique cases.
The remaining \PctFiveOrMoreEq{} of cases have five or seven equilibria.
In all but one case, the solver successfully returned one or more
equilibria, and in all successful cases the number of equilibria
was odd.%
\footnote{Equilibrium existence is guaranteed by Theorem~1 of
\cite{ctgames} under Assumptions~\ref{assp:dx}--\ref{assp:pi},
so the single case with no equilibrium found represents an
incomplete search due to numerical issues encountered with
extreme parameter values.}
When we restrict attention to economically relevant parameter values
where monopoly is viable but duopoly is unprofitable
($\thetaM + \thetaC < 0$), firms discount at standard rates
($\rho \leq 0.05$), and moves are not too infrequent relative to
discounting ($\lambda \geq 0.1$, so $\lambda \geq 2\rho$), the
multiplicity rate increases to \MultShareRelevant{}.
Instead of focusing on the overall uniqueness rate, below we will
characterize where in the parameter space multiplicity is likely
to occur.

\subsection{Evaluation of the Theoretical Conditions}

Using the numerical results, we evaluate the sufficient
conditions from Theorem~\ref{thm:contraction} and
Theorem~\ref{thm:symmetric-uniqueness}.
Since the numerical model uses the logit specification of
Assumption~\ref{assp:ev1}, we check the conditions with
$\cF = \cFLogit{}$ and $\cFp = \cFpLogit{}$
(Lemma~\ref{lem:logit-lipschitz}).
To allow for the possibility that uniqueness extends beyond the
regions characterized by the theorems, we evaluate each condition
at several thresholds $\tau$.
In the case of Theorem~\ref{thm:contraction}, we can restate the
condition as $\lambarminus (U + \lambar \Psi) / (2\rho^2) < \tau$,
where $\tau = 1$ is equivalent to the condition of the theorem and
$\tau > 1$ represents a relaxation.
For Theorem~\ref{thm:symmetric-uniqueness}, we apply the contraction
inequality~\eqref{eq:symmetric-contraction} exactly and instead relax
the switching-cost requirement of
Assumption~\ref{assp:symmetric-costs}.
The assumption requires $\thetaEC + \thetaSV = 0$, so we consider
relaxations $|\thetaEC + \thetaSV| < \tau$ with $\tau > 0$ to allow
switching costs that do not offset exactly.
Note that the threshold $\tau$ is not comparable across the two cases:
for Theorem~\ref{thm:contraction} it scales the sufficient inequality,
while for Theorem~\ref{thm:symmetric-uniqueness} it bounds the
deviation from offsetting switching costs in payoff units.

\begin{table}[t]
\centering
\caption{Coverage of the sufficient conditions for uniqueness}
\label{tab:verification}
\begin{tabular}{lrrrr}
\toprule
Threshold $\tau$ & Predicts Unique & Precision & False Positives & Recall \\
\midrule
\multicolumn{5}{l}{\emph{Panel A. Theorem~\ref{thm:contraction}: $\lambarminus (U + \lambar \Psi) / (2\rho^2) < \tau$}} \\
\midrule
1 & 999{,}421 & 100.0\% & 0 & 1.0\% \\
2 & 1{,}813{,}168 & 100.0\% & 0 & 1.9\% \\
5 & 3{,}782{,}160 & 99.8\% & 6{,}067 & 3.9\% \\
10 & 6{,}341{,}101 & 99.0\% & 60{,}365 & 6.5\% \\
20 & 10{,}293{,}998 & 97.9\% & 216{,}434 & 10.4\% \\
\midrule
\multicolumn{5}{l}{\emph{Panel B. Theorem~\ref{thm:symmetric-uniqueness}: \eqref{eq:symmetric-contraction} and $|\thetaEC + \thetaSV| < \tau$}} \\
\midrule
0.2 & 7{,}414{,}923 & 100.0\% & 0 & 7.7\% \\
0.5 & 17{,}969{,}963 & 99.9\% & 36 & 18.6\% \\
1.0 & 34{,}048{,}658 & 99.9\% & 13{,}519 & 35.2\% \\
2.0 & 60{,}531{,}716 & 99.6\% & 213{,}733 & 62.4\% \\
3.0 & 79{,}447{,}837 & 99.1\% & 727{,}282 & 81.4\% \\
$\infty$ & 94{,}580{,}777 & 98.1\% & 1{,}827{,}215 & 95.9\% \\
\bottomrule
\end{tabular}
\begin{flushleft}
\small \textit{Notes:} Precision is the fraction of cases predicted unique
that are in fact unique. Recall is the fraction of all unique cases that
the condition identifies.
Precision is shown as 100.0\% only when there are no false positives.
In Panel~A, $\tau = 1$ is the exact sufficient condition of
Theorem~\ref{thm:contraction}, and $\tau > 1$ relaxes it by the factor
$\tau$.
In Panel~B, we impose the requirement~\eqref{eq:symmetric-contraction}
of Theorem~\ref{thm:symmetric-uniqueness} exactly while $\tau$ relaxes
the switching-cost requirement $\thetaEC + \thetaSV = 0$.
\end{flushleft}
\end{table}

The results are shown in Table~\ref{tab:verification}, which classifies
predictions as ``positive'' or ``negative.''
A positive prediction here means that the specified uniqueness condition
is satisfied.
A negative prediction, on the other hand, simply means that the condition
is false but does not imply that uniqueness fails.
Precision is defined as $\text{TP}/(\text{TP} + \text{FP})$, where TP
and FP are the numbers of true positives and false positives, respectively.
Precision is thus the fraction of positive predictions that in fact have a
unique equilibrium.
Recall is defined as $\text{TP}/(\text{TP} + \text{FN})$, the fraction
of unique cases captured by the condition, where FN denotes the number
of false negatives.
Since the conditions are sufficient for uniqueness, false positives
should not occur, and we expect precision to be 100\%.
Recall instead measures how much of the overall uniqueness region the
condition includes.

There are no false positives at our smallest thresholds, so the
numerical results are consistent with both theorems.
For Theorem~\ref{thm:contraction}, we encounter the first false
positive at $\tau = \ThmOneFrontier{}$, so the lower bound on
$\rho^2$ is conservative by roughly a factor of two in our search.
The theorem imposes no special structure on payoffs and applies to
games with any finite number of players, any finite state space, and
any bounded error-difference density, but the cost is that the
coverage is limited to \ThmOneRecall{} of all unique-equilibrium
cases.
Theorem~\ref{thm:symmetric-uniqueness} trades off generality for a
condition that does not require strong discounting.
For $\tau = 0.2$, we retain perfect precision and include
\ThmTwoRecallOneFifth{} of all unique cases.
The first false positive appears at $\tau = \ThmTwoFrontier{}$.
With $\tau = 3$, the condition identifies \ThmTwoRecallThree{}
of all unique-equilibrium cases, but at the expense of more false
positives.

The false positives that appear as Assumption~\ref{assp:symmetric-costs}
is relaxed, and that remain when it is dropped entirely ($\tau = \infty$),
show that the inequality from Theorem~\ref{thm:symmetric-uniqueness} is
not by itself sufficient for uniqueness.
However, near the exactly offsetting case it appears to be close to sharp.
The inequality~\eqref{eq:symmetric-contraction} holds at
\ThmTwoSliceHolds{} of the \ThmTwoSliceN{} vectors with
$|\thetaEC + \thetaSV| < 0.2$, and at every vector in this slice with
multiple equilibria the left-hand side is at least
\ThmTwoSliceFrontier{}, so
the inequality could not be relaxed far without introducing false
positives.
The low recall at small $\tau$ in Panel~B of Table~\ref{tab:verification}
therefore reflects the small number of sampled vectors with nearly
offsetting switching costs, rather than looseness in the inequality.
This suggests that a version of Theorem~\ref{thm:symmetric-uniqueness}
with slack in the restriction $\thetaEC + \thetaSV = 0$ may hold.

\subsection{Determinants of Multiplicity}

\begin{figure}[t]
\centering
\includegraphics[width=\textwidth]{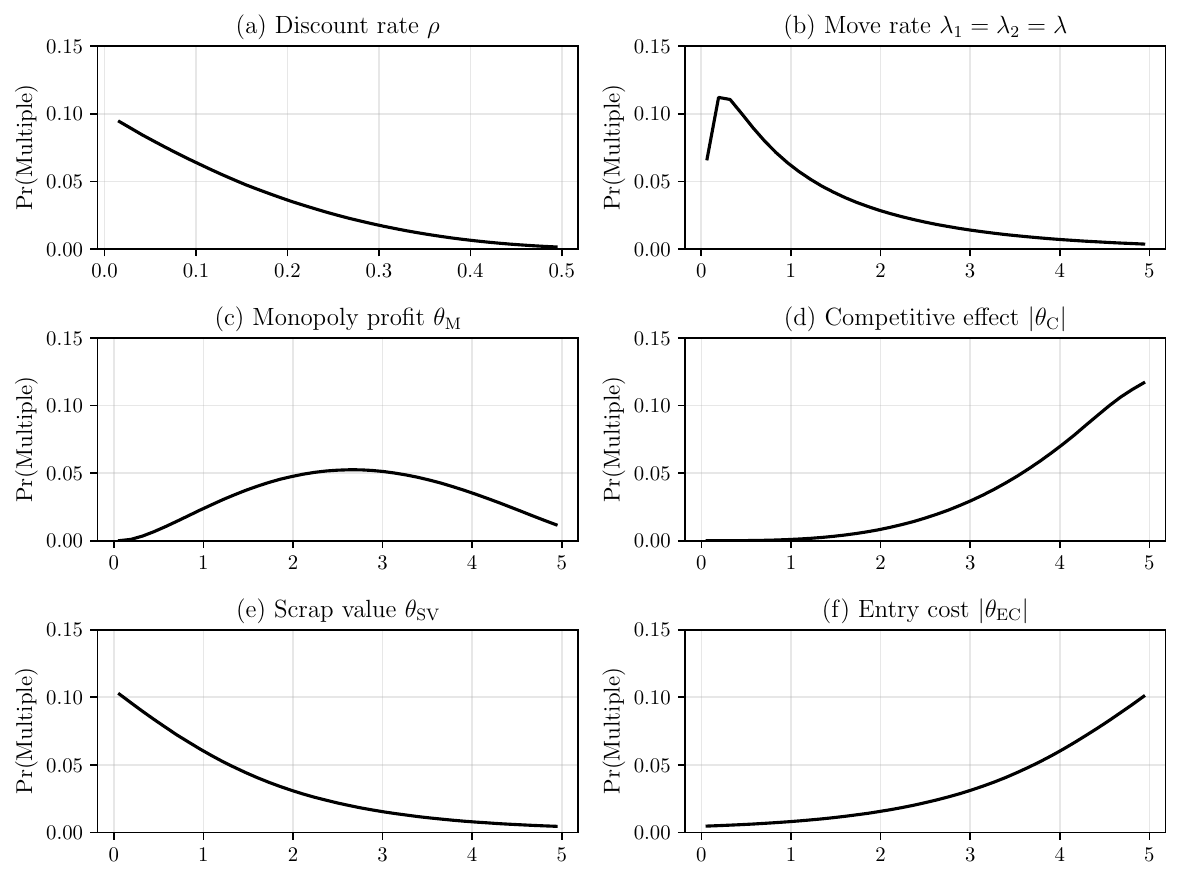}
\caption{Empirical probability of multiple equilibria as a function of
model parameters, averaged over all other parameters.}
\label{fig:profiles}
\end{figure}

The numerical results exhibit several patterns in the parameters
and underlying economic motives that determine multiplicity.
Figure~\ref{fig:profiles} provides a marginal summary of how each
model parameter affects multiplicity.
Overall, multiplicity decreases with the discount rate $\rho$~(a) and
scrap value $\thetaSV$~(e).
Multiplicity increases with the competitive effect $|\thetaC|$~(d) and
entry cost $|\thetaEC|$~(f).
The effect of the move arrival rate $\lambda$~(b) is nonmonotone.
Multiplicity is more likely to occur at intermediate values of
$\lambda$ than at extreme values due to a mechanism we return to below.
The monopoly flow profit $\thetaM$~(c) gradually increases and then
decreases multiplicity.
When $\thetaM$ is small, firms have little incentive to enter, but
when $\thetaM$ is large, duopoly becomes viable
($\thetaM + \thetaC > 0$) and both firms find it profitable to enter.

In Figure~\ref{fig:determinants}, we examine the joint relationships
between various structural parameters.
We focus on three determinants of multiplicity: patience and strategic
interaction, competition intensity, and switching costs.

\paragraph{Patience and strategic interaction.}
Figure~\ref{fig:determinants}(a) shows the probability of multiple
equilibria as a function of the discount rate $\rho$ and the move
arrival rate $\lambda$.
Multiplicity is concentrated in the lower-left corner, where firms are
patient and moves are infrequent.
As expected from the theoretical results, and previously shown in the
marginal plot (Figure~\ref{fig:profiles}(b)), the effect of $\lambda$
is nonmonotone.
Increasing $\lambda$ has two effects: it increases strategic
interaction, so firms interact more frequently, but it also decreases
firms' ability to commit to actions, since each firm expects
the other to re-evaluate and respond more frequently.
In the limit, as $\lambda \to \infty$, switching costs dominate
flow payoffs and firms expect rival states to persist only briefly,
so best responses become nearly independent of rival states.
Under offsetting switching costs, Corollary~\ref{cor:discount-free}
shows that the equilibrium is then unique.
Therefore, multiplicity is highest at intermediate values of
$\lambda$ while uniqueness is more likely for very low or very
high values of $\lambda$.
As a result, the multiplicity contours form arcs whose peaks
shift to the right as $\rho$ increases.

\begin{figure}[t]
\centering
\includegraphics[width=\textwidth]{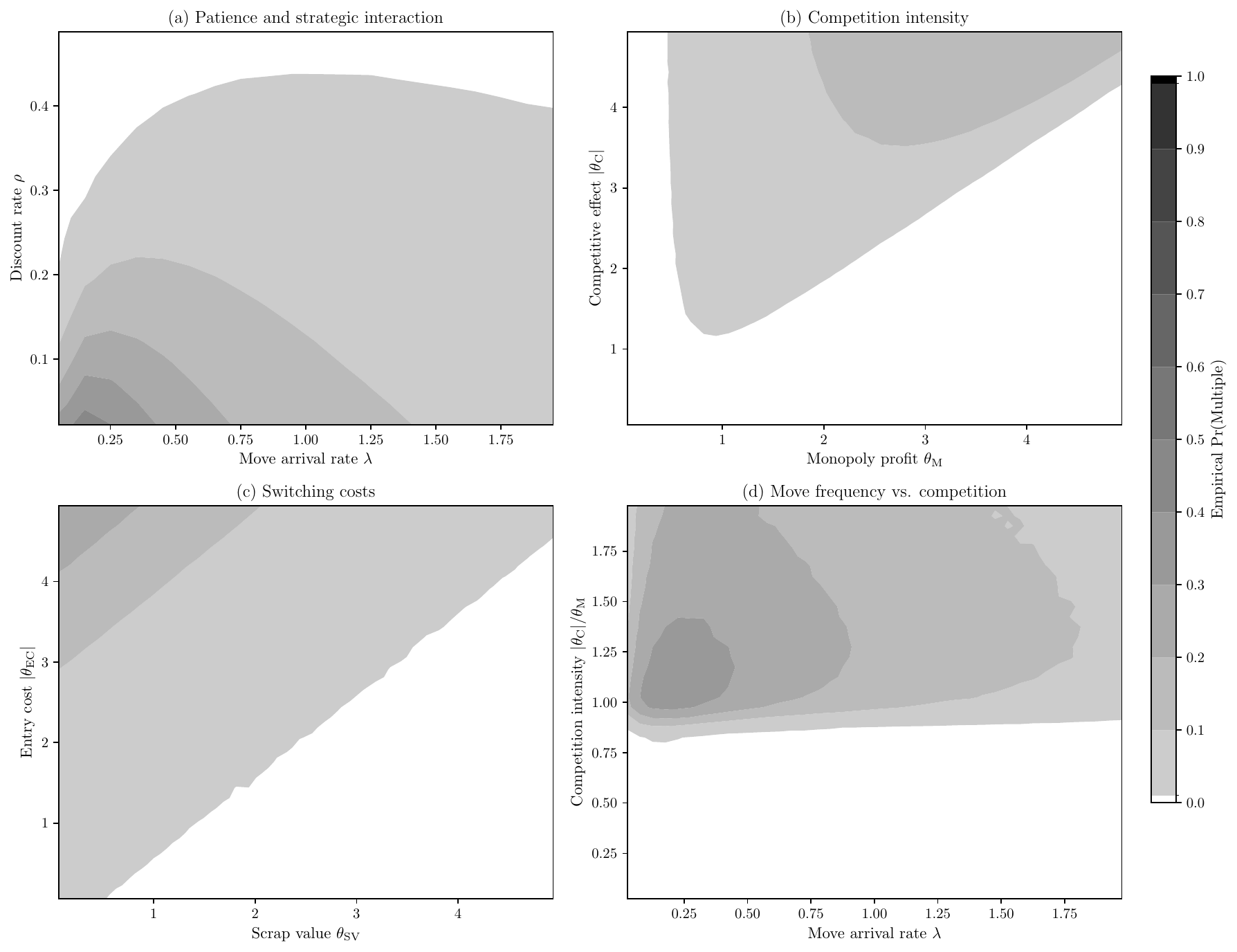}
\caption{Empirical probability of multiple equilibria as a function of
important model parameter pairs.
Shading indicates the fraction of parameter vectors in each bin with
multiple equilibria, from white (all unique) to black (all multiple).
Based on \NSample{} parameter vectors from a Halton sequence with
$\lambda_1 = \lambda_2$.}
\label{fig:determinants}
\end{figure}

\paragraph{Competition intensity.}
Figure~\ref{fig:determinants}(b) shows the multiplicity probability
as a function of $|\thetaC|$ and $\thetaM$.
The ratio of these parameters is a measure of competition intensity,
and indeed there is a sharp threshold just below the diagonal, which
represents the boundary between profitable ($|\thetaC| / \thetaM < 1$)
and unprofitable ($|\thetaC| / \thetaM > 1$) duopoly.
When duopoly is profitable, it is more likely for the model to admit a
unique equilibrium where both firms are active.
Otherwise, the market cannot support both firms, creating a coordination
problem that gives rise to multiple asymmetric equilibria.
These equilibria have a self-fulfilling character, where each firm
stays active only when it expects its rival to exit.
We see the same pattern in Figure~\ref{fig:determinants}(d), which
plots the competition intensity ratio against $\lambda$.
Multiplicity is rare when $|\thetaC|/\thetaM \ll 1$ regardless of
the move arrival rate.

\paragraph{Switching costs.}
Figure~\ref{fig:determinants}(c) shows how multiplicity relates to
switching costs.
Multiplicity is most likely in the upper left (high entry cost, low
scrap value) and least likely in the lower right (low entry cost,
high scrap value), with clear contour lines indicating that the
probability of multiplicity depends heavily on the net switching
cost $\thetaEC + \thetaSV$.
High scrap values induce exit, which reduces the coordination problem
discussed above, while high entry costs exacerbate it.
This also highlights the role of switching costs in commitment.
When firms have to pay high entry costs and recover little value upon
exit (upper-left corner), state transitions are costly in
both directions.
Market structure adjusts slowly, so each firm's current status serves
as a commitment, and as a result, it is easier to support multiple
equilibria.
On the other hand, when firms can enter easily and recover large
values when they exit (lower-right corner), switching is
relatively costless and commitment is weak, making it more difficult
to support multiple equilibria.

\paragraph{Distributions of payoff-related quantities.}
Figure~\ref{fig:histograms} displays the densities of four
payoff-related quantities conditional on unique-equilibrium and
multiple-equilibria cases.
The dark gray densities are based on the \PctMultiple{}
of parameter vectors that exhibit multiple equilibria.
First, the monopoly profit $\thetaM$~(a) is nearly uniformly
distributed among unique-equilibrium cases, but is unimodal
under multiplicity.
The duopoly profit $\thetaM + \thetaC$~(b) shows more separation.
The distribution for multiple-equilibria cases is concentrated
at negative values (unprofitable duopoly) while the unique-equilibrium
distribution resembles the unconditional triangular distribution (since
duopoly profit is the sum of two uniformly distributed variables whose
supports have equal widths).
The distribution of the competition intensity ratio
$|\thetaC|/\thetaM$~(c) for multiple-equilibria cases is concentrated
just above the threshold near one, while the density for
unique-equilibrium cases is highest below one.
Finally, the distribution of the net switching cost
$\thetaEC + \thetaSV$~(d) is shifted toward negative values for
multiple-equilibria cases,
where entry costs exceed scrap values (increasing commitment).
The distribution conditional on uniqueness is again nearly triangular.

\begin{figure}[t]
\centering
\includegraphics[width=\textwidth]{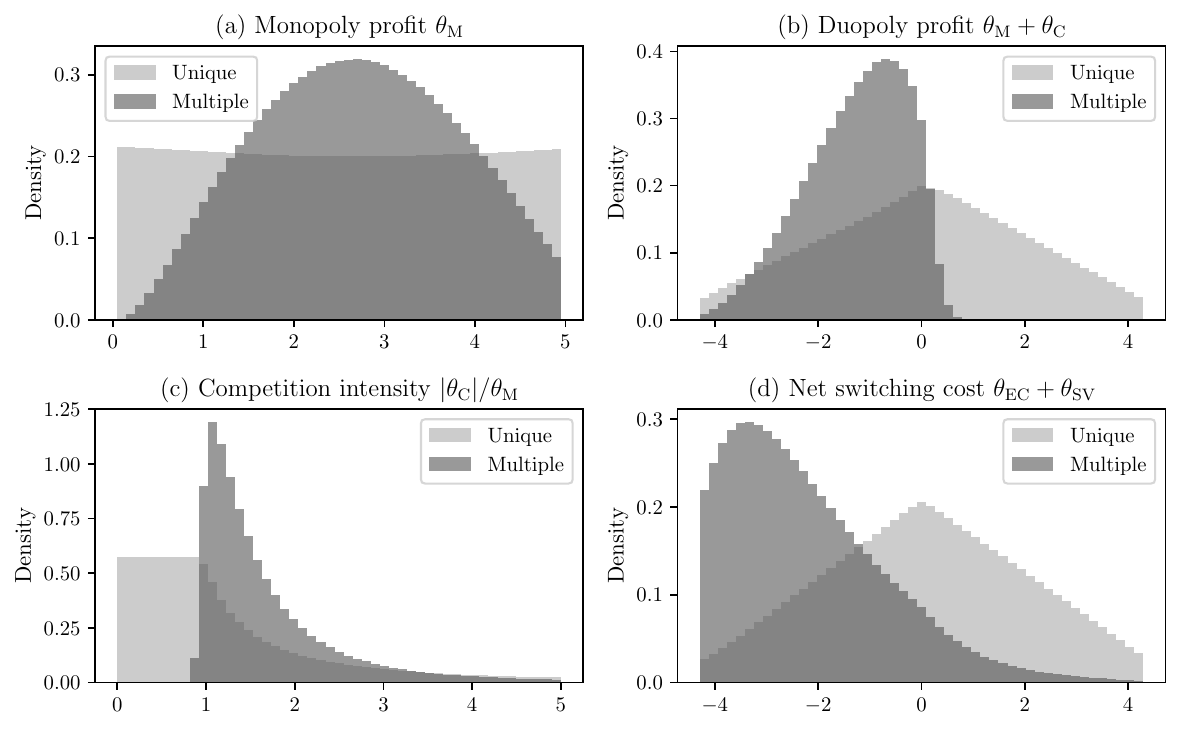}
\caption{Conditional densities of payoff-related quantities by equilibrium
multiplicity. Each panel shows the distribution of a quantity among
parameter vectors with a unique equilibrium (light) and among those
with multiple equilibria (dark).}
\label{fig:histograms}
\end{figure}

\subsection{Equilibrium Selection and Long-Run Outcomes}
\label{sec:differences}

In this section we quantify how the equilibria differ economically
in three ways: in long-run market structure, in producer surplus,
and in which firm occupies the market.
Each distinct equilibrium induces a continuous-time Markov chain
over market states $x_k$.
For each equilibrium at each of the \NMultiSeeds{} parameter vectors
where we found multiple equilibria, we compute the stationary
distribution along with two measures of producer surplus:
the stationary flow of operating profits and the present value
$V_{11} + V_{21}$ in state $x_1 = (0,0)$.
The latter includes switching costs and the expected surplus
from the choice-specific errors.
Table~\ref{tab:outcomes} reports the distribution of spreads:
the maximum minus the minimum over equilibria at each parameter
vector.

\begin{table}[t]
\centering
\caption{Spreads of long-run outcomes across equilibria}
\label{tab:outcomes}
\begin{tabular}{lrrrrrr}
\toprule
& \multicolumn{3}{c}{All Equilibria} & \multicolumn{3}{c}{Stable Pair} \\
\cmidrule(lr){2-4} \cmidrule(lr){5-7}
Long-Run Outcome & Median & 90\% & Max & Median & 90\% & Max \\
\midrule
\multicolumn{7}{l}{\emph{Aggregate outcomes}} \\
Duopoly share & 0.128 & 0.337 & 0.892 & 0.000 & 0.000 & 0.000 \\
Monopoly share & 0.125 & 0.336 & 0.892 & 0.000 & 0.000 & 0.000 \\
Empty-market share & 0.001 & 0.011 & 0.099 & 0.000 & 0.000 & 0.000 \\
Producer surplus, flow $/\,U$ & 0.164 & 0.406 & 0.897 & 0.000 & 0.000 & 0.000 \\
Producer surplus, PDV $/\,(U/\rho)$ & 0.063 & 0.250 & 0.746 & 0.000 & 0.000 & 0.000 \\
\addlinespace
\multicolumn{7}{l}{\emph{Firm identity}} \\
Firm 1 active share & 0.811 & 0.988 & 1.000 & 0.810 & 0.987 & 1.000 \\
Firm 1 value, PDV $/\,(U/\rho)$ & 0.238 & 0.468 & 0.859 & 0.238 & 0.469 & 0.859 \\
\bottomrule
\end{tabular}
\begin{flushleft}
\small \textit{Notes:} Spread (maximum minus minimum) of each
outcome across equilibria at each of the 3{,}305{,}491 parameter vectors with multiple equilibria.
Stable Pair columns condition on the 3{,}288{,}866 vectors with exactly two stable equilibria.
Surplus and value spreads are expressed in units of the
span of flow payoffs $U$ or the present value $U/\rho$,
so that entries are comparable across parameter vectors.
PDV denotes present discounted value.
\end{flushleft}
\end{table}

The long-run probability of duopoly differs across all equilibria,
including the unstable one, by \MedianDuoSpread{} at the median
across parameter vectors with multiple equilibria, and by \NinetyDuoSpread{}
at the 90th percentile.
Since the shares sum to one, the duopoly and monopoly spreads can
differ by at most the empty-market spread.
The latter is small for most parameter vectors, resulting in
the nearly identical marginal distributions in Table~\ref{tab:outcomes}.
What equilibrium selection changes, therefore, is how the stationary
distribution splits between duopoly and monopoly, not whether
the market is served.
Equilibrium selection also affects total producer surplus: the flow
measure differs by \MedianPSFlowSpread{}
at the median (normalized by the span of flow payoffs $U$).

Across all sampled parameter vectors, we found that \PctUniqueStable{}
of unique equilibria are stable.
Among parameter vectors that produce multiple equilibria,
\PctCanonicalStability{} exhibit a standard configuration with
three equilibria, two stable and one unstable.
Table~\ref{tab:stability} displays the full distribution of
equilibrium counts and stability among multiple-equilibria cases.%
\footnote{The \ZeroStablePct{} of multiple-equilibria cases with
three unstable equilibria have the same structure as the standard
configuration, one symmetric equilibrium and a mirror-image pair,
but the paired equilibria are no longer stable.
The smallest spectral radius among the three equilibria is
\ZeroStableMedRadius{} at the median.
These cases are characterized by low discount and move arrival rates
(median $\rho = \ZeroStableMedRho{}$ and
$\lambda = \ZeroStableMedLambda{}$)
and strong competition effects (median
$|\thetaC|/\thetaM = \ZeroStableMedComp{}$).}
Stability is important for both computation, since only stable
equilibria can be reached by iterating $T$ from all nearby starting
values, and estimation, where an analogous property of the iterative
continuous-time nested pseudo-likelihood (CTNPL) estimator determines
its convergence behavior.

\begin{table}[t]
\centering
\caption{Stability of multiple equilibria}
\label{tab:stability}
\begin{tabular}{ccccrr}
\toprule
\multicolumn{2}{c}{All Equilibria} & \multicolumn{2}{c}{Symmetric Equilibria} & \multicolumn{2}{c}{Occurrence} \\
\cmidrule(lr){1-2}\cmidrule(lr){3-4}\cmidrule(lr){5-6}
Total & Stable & Stable & Unstable & Count & Percent \\
\midrule
3 & 0 & 0 & 1 & 14{,}878 & 0.45\% \\
3 & 2 & 0 & 1 & 3{,}288{,}717 & 99.49\% \\
5 & 2 & 0 & 1 & 144 & 0.0044\% \\
5 & 3 & 1 & 0 & 1{,}740 & 0.05\% \\
7 & 2 & 0 & 1 & 5 & 0.0002\% \\
7 & 4 & 0 & 1 & 7 & 0.0002\% \\
\bottomrule
\end{tabular}
\end{table}

Within the standard three-equilibrium configuration, selection between
the two stable equilibria is primarily distributional in our symmetric
model: it determines which firm occupies the market, but not the
distributions of aggregate market structure or total producer surplus.
The two stable equilibria are mirror images with swapped firm roles,
so the aggregate-outcome entries in the stable-pair columns of
Table~\ref{tab:outcomes} are identically zero.
Firm~1's stationary active share, on the other hand, differs across
equilibria by \MedianFOneSpread{} at the median
and \NinetyFOneSpread{} at the 90th percentile.
Being the firm with the higher activity probability is worth
\MedianFirmValueSpread{} at the median in present discounted value
terms (normalized by $U/\rho$).

\section{Implications for Estimation and Counterfactuals}
\label{sec:estimation}

There are well-known consequences of multiplicity for applied
researchers: estimation is complicated by the lack of a well-defined
reduced form, and counterfactual predictions depend on equilibrium
selection.
In this section, we quantify both of these using the equilibria
computed in Section~\ref{sec:numerical} in combination with a targeted
Monte Carlo experiment and a simple counterfactual exercise.

\subsection{Estimation under Multiplicity}
\label{sec:estimation-mc}

We consider maximum likelihood estimation (MLE) of
$\theta = (\thetaM, \thetaC, \thetaEC)$ using data on the
states and actions observed in \MCMarkets{} independent markets over a
time interval $[0, \MCHorizon{}]$.
Following the literature, we take $(\rho, \lambda)$ as known, and
we fix the scrap value $\thetaSV$ at its true value to normalize
payoffs for identification
\citep{magnac-thesmar-2002, aguirregabiria-suzuki-2014}.%
\footnote{For any $b$, replacing $(\thetaM, \thetaEC, \thetaSV)$ by
$(\thetaM + \rho b, \thetaEC - b, \thetaSV + b)$ shifts each firm's
value function by $b$ in the states where it is active and leaves
all equilibrium choice probabilities unchanged.
Fixing $\thetaSV$ removes this indeterminacy.}

Four data-generating processes (DGPs) are specified along
a $\thetaC$ slice that passes through a bifurcation in the
equilibrium mapping.%
\footnote{We assume all markets play the same equilibrium,
allowing data to be pooled.
\citet{otsu-pesendorfer-takahashi-2016} consider testing
whether different markets play different equilibria.}
Two DGPs are on the unique-equilibrium side of the bifurcation,
and the remaining two share a single multiple-equilibria parameter
vector, but one generates data from a stable equilibrium and the
other from the unstable equilibrium (the stable and unstable DGPs,
respectively).

\begin{figure}[t]
\centering
\includegraphics[width=0.85\textwidth]{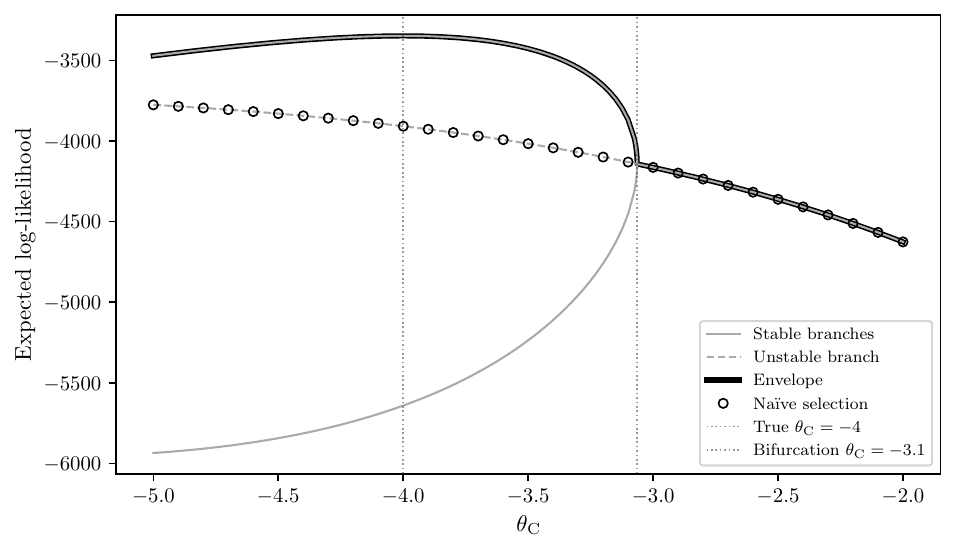}
\caption{Expected log-likelihood of each equilibrium branch
(stable branches solid gray, unstable branch dashed)
along a $\thetaC$ slice through a bifurcation point,
for \MCMarkets{} markets observed over $[0, \MCHorizon{}]$ under the
stable DGP.
The envelope over branches (solid black, oracle MLE objective) is
maximized at the true parameter value.
Circles mark the branch selected by the solver's default starting
value (na{\"i}ve MLE objective).}
\label{fig:loglik}
\end{figure}

The four DGPs can be understood visually in Figure~\ref{fig:loglik},
which plots the expected log-likelihood of each equilibrium branch
along the $\thetaC$ slice, with the data fixed at the
stable DGP (true $\thetaC = \MCThetaCMultiple{}$).
Data are simulated from the equilibrium CCPs corresponding to the top
branch (solid gray line).
The remaining parameters are held fixed at true values
$(\thetaM, \thetaEC, \thetaSV, \rho, \lambda) = (2, -2, 1, 0.05, 1)$.
The unstable DGP has the same true
$\thetaC = \MCThetaCMultiple{}$, but data are simulated instead from
the CCPs corresponding to the unstable branch (dashed gray line).
The remaining two DGPs have true parameters on the unique-equilibrium
side of the bifurcation at $\thetaC = \MCThetaCBifurcation{}$, with
one intentionally near the branching point
($\thetaC = \MCThetaCNear{}$, the near-boundary DGP) and the other
far inside the uniqueness region ($\thetaC = \MCThetaCFar{}$, the
interior DGP).
Neither theorem covers the DGP parameters, so uniqueness at the
near-boundary and interior DGPs was verified numerically.%
\footnote{We selected the DGP to be representative: a standard
discount rate, $\rho = 0.05$, flow payoffs under which monopoly is
profitable but duopoly is not, and a spread in long-run duopoly share
across equilibria near the median among multiple-equilibria cases.}

\begin{table}[tp]
\centering
\caption{Estimation under multiplicity: Monte Carlo}
\label{tab:mc}
\small
\setlength{\tabcolsep}{4pt}
\begin{tabular}{lrrrr}
\toprule
& \multicolumn{3}{c}{Bias (RMSE)} & Selects \\
\cmidrule(lr){2-4}
Estimator & $\thetaM$ & $\thetaC$ & $\thetaEC$ & DGP Eq.\ \\
\midrule
\multicolumn{5}{l}{\emph{Unique (interior): $\thetaC = -2.00$, $\rho_{\mathrm{DT}} = 0.976$, $\rho_{\mathrm{NPL}} = 0.478$}} \\
Na{\"i}ve MLE & $+0.005$ ($0.099$) & $-0.004$ ($0.134$) & $+0.000$ ($0.067$) & 100.00\% \\
Oracle MLE & $+0.005$ ($0.099$) & $-0.004$ ($0.134$) & $+0.000$ ($0.067$) & 100.00\% \\
CTNPL (one iter.) & $-0.002$ ($0.103$) & $+0.004$ ($0.138$) & $-0.001$ ($0.067$) & 100.00\% \\
CTNPL (converged) & $+0.005$ ($0.099$) & $-0.004$ ($0.134$) & $+0.000$ ($0.067$) & 100.00\% \\
\addlinespace
\multicolumn{5}{l}{\emph{Unique (near boundary): $\thetaC = -3.00$, $\rho_{\mathrm{DT}} = 0.990$, $\rho_{\mathrm{NPL}} = 0.968$}} \\
Na{\"i}ve MLE & $+0.002$ ($0.082$) & $-0.006$ ($0.148$) & $-0.002$ ($0.064$) & 100.00\% \\
Oracle MLE & $+0.002$ ($0.082$) & $-0.006$ ($0.148$) & $-0.002$ ($0.064$) & 100.00\% \\
CTNPL (one iter.) & $-0.007$ ($0.089$) & $+0.010$ ($0.159$) & $-0.003$ ($0.064$) & 100.00\% \\
CTNPL (converged) & $-0.002$ ($0.079$) & $+0.011$ ($0.130$) & $-0.001$ ($0.064$) & 100.00\% \\
\addlinespace
\multicolumn{5}{l}{\emph{Multiple (stable DGP): $\thetaC = -4.00$, $\rho_{\mathrm{DT}} = 0.976$, $\rho_{\mathrm{NPL}} = 0.065$}} \\
Na{\"i}ve MLE & $-0.904$ ($0.905$) & $+0.334$ ($0.395$) & $-0.882$ ($0.887$) & 0.00\% \\
Oracle MLE & $-0.000$ ($0.101$) & $+0.008$ ($0.220$) & $-0.009$ ($0.119$) & 100.00\% \\
CTNPL (one iter.) & $-0.017$ ($0.110$) & $+0.042$ ($0.240$) & $-0.018$ ($0.135$) & 99.60\% \\
CTNPL (converged) & $-0.001$ ($0.110$) & $+0.009$ ($0.232$) & $-0.009$ ($0.125$) & 100.00\% \\
\addlinespace
\multicolumn{5}{l}{\emph{Multiple (unstable DGP): $\thetaC = -4.00$, $\rho_{\mathrm{DT}} = 1.126$, $\rho_{\mathrm{NPL}} = 1.463$}} \\
Na{\"i}ve MLE & $-0.000$ ($0.079$) & $-0.008$ ($0.197$) & $-0.002$ ($0.069$) & 100.00\% \\
Oracle MLE & $-0.000$ ($0.079$) & $-0.008$ ($0.197$) & $-0.002$ ($0.069$) & 100.00\% \\
CTNPL (one iter.) & $-0.013$ ($0.087$) & $+0.021$ ($0.210$) & $-0.003$ ($0.069$) & 100.00\% \\
CTNPL (converged) & $-0.161$ ($0.175$) & $+0.789$ ($0.797$) & $+0.016$ ($0.073$) & 0.00\% \\
\bottomrule
\end{tabular}
\begin{flushleft}
\small \textit{Notes:} Bias (RMSE in parentheses) over 500 replications; 200 markets observed continuously over $[0, 25]$;
$(\rho, \lambda)$ known and $\thetaSV$ fixed at its
true value; optimizers start at the
true parameters. The DGPs vary
$\thetaC$ along a slice through a multiple-equilibria parameter
vector (multiplicity boundary at $\thetaC \approx -3.1$).
Section~\ref{sec:estimation-mc} defines the estimators. At the
stable DGP the na{\"i}ve solver's default start selects an
equilibrium 0.31 from the DGP equilibrium in maximum
CCP distance. The final column reports the share of replications
in which the estimated CCPs lie within
0.1 in maximum absolute CCP distance
of the DGP equilibrium.
$\rho_{\mathrm{DT}}$ and $\rho_{\mathrm{NPL}}$ are the spectral radii of
the Jacobian $DT(V^*)$ and the CTNPL operator at the DGP
equilibrium, respectively.
\end{flushleft}
\end{table}

We carry out a Monte Carlo experiment that compares three estimators
across the four DGPs.
First, we consider a na{\"i}ve nested fixed point MLE
\citep{rust87optimal} that assumes uniqueness.
For each trial parameter vector, this estimator finds a single
equilibrium by applying Newton's method to the equilibrium system,
initialized at
$V_{0} = [0,\,0,\,\thetaM/\rho,\,(\thetaM+\thetaC)/\rho,\,0,\,\thetaM/\rho,\,0,\,(\thetaM+\thetaC)/\rho]^\top$,
the vector of perpetuity values $u_{ik}/\rho$, with components
(stacked across players) corresponding to the present discounted
values of monopoly and duopoly flows in the respective states.
Another default a researcher might choose is the zero vector.
Both are symmetric across firms, a point we return to below.

The second estimator is the oracle MLE, which serves as an infeasible
benchmark: for each parameter vector we ``warm-start'' the Newton
solver at the true DGP equilibrium's value function.
The third estimator is the CTNPL estimator of
\citet{blevins-kim-2024}, an iterative estimator
that solves a linear valuation system given the current choice
probabilities at each step, maximizes the implied pseudo-likelihood
over $\theta$, and then updates the CCPs using the best-response
mapping.
We initialize this estimator at nonparametric frequency estimates
of the CCPs, which consistently estimate those of the equilibrium
that generated the data without requiring the researcher
to know it in advance.
We compare CTNPL after a single iteration (the ``two-step''
estimator) and after iterating to convergence.%
\footnote{We terminate the iteration when either the maximum change
in the CCPs falls below $10^{\NPLTolExp}$ or after \NPLMaxStages{}
stages.
All replications converged except in the near-boundary DGP, where
\CTNPLNearConverged{} converged (median \CTNPLNearMedianIter{}
iterations); replications that did not converge are included in
Table~\ref{tab:mc} at their estimates after \NPLMaxStages{} stages.}

In Table~\ref{tab:mc}, we compare the estimates over \MCReps{}
Monte Carlo replications and report the bias and root mean squared
error (RMSE) of each estimator across the four DGPs.
In all cases, we initialize the optimizer at the corresponding true
parameter vector to separate equilibrium selection from optimization
error;
however, the initializing equilibria (value functions or CCPs)
differ across estimators as previously described.
The final column of Table~\ref{tab:mc} reports the percentage of
replications for which the estimator selects the correct DGP
equilibrium.

As expected, the oracle MLE tracks the correct equilibrium branch as
the parameters vary near the truth.
This estimator maximizes the sample counterpart of the upper envelope
in Figure~\ref{fig:loglik}, which is maximized at the true parameter
vector.
Conditional on correct equilibrium selection, the identified
parameters are estimated precisely.
The gap between the oracle MLE and other estimators therefore
represents the equilibrium selection error.

The na{\"i}ve MLE, however, assumes there is always a unique
equilibrium and maximizes the likelihood of whichever equilibrium its
inner solver returns.
As expected, this estimator performs well under both
unique-equilibrium DGPs.
The results for the multiple-equilibria DGPs are somewhat
counterintuitive: the na{\"i}ve MLE performs well under the
unstable DGP but has large biases under the stable DGP.

The reason is that Newton's method solves the equilibrium system
directly, so the equilibrium it returns depends on the starting value
rather than on stability under iteration of $T$.
If $\lambda_1 = \lambda_2$, the equilibrium system commutes with the
firm-relabeling map $\Xi$ from the proof of
Theorem~\ref{thm:symmetric-uniqueness} (which does not depend
on Assumption~\ref{assp:symmetric-costs}), and so do its Jacobian
and each Newton step.
Thus the iterates generated by Newton's method when initialized at a
symmetric value function must be symmetric and can only converge to a
symmetric equilibrium.
Since this slice has only one symmetric equilibrium
branch---the unstable branch---the na{\"i}ve MLE
tracks it, as shown by the circles in
Figure~\ref{fig:loglik}.
The na{\"i}ve objective coincides with the oracle's for the
unstable DGP.
As the final column of Table~\ref{tab:mc} shows, the selection
is correct in 100\% of cases for the unstable DGP and is
consistently wrong for the stable DGP.
The estimator's consistency in the former case therefore hinges
on a coincidence between the equilibrium that generated the data
and the one equilibrium a symmetric initialization can reach.

This dependence on a coincidence is not specific to our DGP.
In the equilibrium sets computed in Section~\ref{sec:numerical},
the parameter vectors with multiple equilibria have exactly one
symmetric equilibrium, which is unstable except in the
five-equilibria cases with three stable equilibria
(Table~\ref{tab:stability}).
Outside these cases, the coincidence breaks down when data are
generated by a stable equilibrium.

CTNPL selects the equilibrium from preliminary frequency
estimates of the CCPs.
With one iteration, it achieves bias and RMSE nearly
as low as the oracle MLE under the unique-equilibrium DGPs.
Under the multiple-equilibria DGPs, the results differ substantially
after iteration to convergence.
The relevant diagnostic is not the spectral radius $\rho_{\mathrm{DT}}$, which governs stability
under iteration of $T$, but rather the spectral radius
$\rho_{\mathrm{NPL}}$ of the CTNPL operator, which maps the current
choice probabilities to their update after one iteration, including
the re-estimation of $\theta$.
This is the continuous-time analogue of the discrete-time NPL
convergence condition discussed by \cite{kasahara-shimotsu-2012}.
We report both measures in each DGP panel heading in
Table~\ref{tab:mc} for comparison.%
\footnote{See \cite{aguirregabiria-mira-2002,aguirregabiria-mira-2007},
\cite{kasahara-shimotsu-2008,kasahara-shimotsu-2012},
\cite{pesendorfer-schmidt-dengler-2010},
\cite{egesdal15dynamic},
\cite{aguirregabiria-marcoux-2021},
\cite{bugni-bunting-2021},
\cite{blevins-kim-2024},
and \cite{dearing-blevins-2025},
among others, for discussions of convergence of iterative estimators.}
At all four DGPs, the two radii fall on the same side of unity.%
\footnote{These are distinct operators and their radii need not
agree. $\rho_{\mathrm{DT}}$ is bounded below by the effective discount factor
$(\lambda_1 + \lambda_2)/(\rho + \lambda_1 + \lambda_2) = 0.976$ here,
since adding a constant $c$ to every component of $V$ shifts $T(V)$
by $c\eta/(\rho + \eta)$, so the constant vector is an eigenvector of
$DT$ with eigenvalue $\eta/(\rho + \eta)$.}
Since the CTNPL iteration is contractive
($\rho_{\mathrm{NPL}} = \NPLRadiusStable{}$) in the stable DGP,
iterating to convergence selects the correct equilibrium.
For the unstable DGP, the CTNPL operator is also unstable
($\rho_{\mathrm{NPL}} = \NPLRadiusUnstable{}$), so further iteration
selects the wrong equilibrium and increases bias in the estimated
parameters.
When the equilibrium is unique, both for the interior DGP and even
near the multiplicity boundary (where $\rho_{\mathrm{NPL}} = \NPLRadiusNear{}$),
the converged CTNPL estimates remain on par with the MLE.
In practice, although the researcher does not know which equilibrium
was realized in the DGP, $\rho_{\mathrm{NPL}}$ can be evaluated at
the consistent one-step CTNPL estimates before iterating further.

\subsection{Counterfactual Selection Ambiguity}
\label{sec:counterfactual}

\begin{table}[tp]
\centering
\caption{Selection ambiguity of counterfactual predictions under multiplicity}
\label{tab:counterfactual}
\begin{tabular}{lrrrrrr}
\toprule
Subsidy & Min & 10\% & Median & 90\% & Max & $n$ \\
\midrule
\multicolumn{7}{l}{\emph{Change in duopoly share, all counterfactual eq.}} \\
$+0.25$ & 0.000 & 0.022 & 0.129 & 0.338 & 0.775 & 8{,}641 \\
$+0.5$ & 0.000 & 0.023 & 0.132 & 0.339 & 0.775 & 7{,}340 \\
$+1.0$ & 0.000 & 0.026 & 0.137 & 0.346 & 0.771 & 5{,}268 \\
Offsetting & 0.001 & 0.037 & 0.169 & 0.365 & 0.713 & 1{,}977 \\
\addlinespace
\multicolumn{7}{l}{\emph{Change in duopoly share, stable counterfactual eq.}} \\
$+0.25$ & 0.000 & 0.000 & 0.000 & 0.000 & 0.342 & 8{,}640 \\
$+0.5$ & 0.000 & 0.000 & 0.000 & 0.000 & 0.342 & 7{,}339 \\
$+1.0$ & 0.000 & 0.000 & 0.000 & 0.000 & 0.344 & 5{,}266 \\
Offsetting & 0.000 & 0.000 & 0.000 & 0.000 & 0.000 & 1{,}977 \\
\addlinespace
\multicolumn{7}{l}{\emph{Change in firm 1 active share, stable counterfactual eq.}} \\
$+0.25$ & 0.003 & 0.376 & 0.804 & 0.988 & 1.000 & 8{,}640 \\
$+0.5$ & 0.010 & 0.367 & 0.802 & 0.988 & 1.000 & 7{,}339 \\
$+1.0$ & 0.007 & 0.365 & 0.804 & 0.990 & 1.000 & 5{,}266 \\
Offsetting & 0.039 & 0.378 & 0.818 & 0.995 & 1.000 & 1{,}977 \\
\bottomrule
\end{tabular}
\begin{flushleft}
\small \textit{Notes:} Entry subsidy $\thetaEC \to \thetaEC + s$ at 10{,}000 parameter vectors with multiple equilibria and at least one stable equilibrium.
For each, we select a stable equilibrium as the baseline.
The Offsetting row sets $s = -(\thetaEC + \thetaSV)$ where
this is positive, so that $\thetaEC + s + \thetaSV = 0$.
Each panel reports the spread (maximum minus minimum)
across the counterfactual equilibria of the stated effect. All panels condition
on the presence of multiple counterfactual equilibria, which
results in fewer cases $n$ as the fixed subsidy increases.
The stable panels further require at least two stable
counterfactual equilibria.
\end{flushleft}
\end{table}

Setting aside estimation, equilibrium selection remains
an issue for counterfactual policy analysis.
Even if the structural parameters can be consistently estimated, the
equilibrium played under counterfactual parameters is indeterminate
under multiplicity \citep{aguirregabiria-2012, jun-pinkse-2020}.
To measure the effects of this indeterminacy, we implement an
entry subsidy in the dynamic duopoly model: we set
$\thetaEC \rightarrow \thetaEC + s$ at the first \NCfSeeds{}
parameter vectors in the Halton sequence with multiple equilibria, at
least one of which is stable.
We take the baseline equilibrium to be the stable equilibrium
favoring firm~1 (a labeling normalization) and compute the
full set of counterfactual equilibria following the same
multi-start procedure used in Section~\ref{sec:numerical}.

Even with the baseline equilibrium known, the predicted effect of the
subsidy depends on which counterfactual equilibrium arises in the
market.
The resulting ambiguity can be seen in Table~\ref{tab:counterfactual},
which reports quantiles of the distribution (across parameter vectors)
of spreads of the predicted effects across counterfactual equilibria.
Each spread is the largest difference between the predictions of any
two selection rules.

We highlight three main results.
First, the subsidy itself changes the coordination problem in the
market and reduces multiplicity: the share of parameter vectors
with multiple counterfactual equilibria falls from
\CfMultiShareQuarter{} at $s = 0.25$ to \CfMultiShareFull{} at $s = 1.0$.
Indeed, the subsidy raises the net switching cost
$\thetaEC + \thetaSV$---a key determinant of multiplicity, as shown in
Figure~\ref{fig:determinants}(c).
Second, among stable counterfactual equilibria, aggregate predictions
such as the duopoly share do not depend on selection: because the
stable equilibria typically occur in mirror-image pairs, the
distribution is zero up to the 90th percentile.
The nonzero maximum arises due to the five-equilibria cases with
three stable equilibria.
Third, counterfactual effects involving firm identity---which firm
serves the market---are very sensitive to equilibrium selection
even when attention is restricted to stable equilibria, as
essentially the entire market is reallocated between monopoly
equilibria.

These findings suggest that researchers should carefully consider
possible multiplicity of counterfactual equilibria, especially for
policy questions involving firm identity.
Label-invariant aggregates coincide across stable counterfactual
equilibria in our symmetric model, although this need not hold with
asymmetric firms.

For counterfactual parameters satisfying the conditions of
Theorem~\ref{thm:contraction} or
Theorem~\ref{thm:symmetric-uniqueness}, the analysis can proceed
without ambiguity.
To illustrate, consider the offsetting subsidy $s = -(\thetaEC + \thetaSV)$,
which makes entry costs fully recoverable upon exit, as in
Assumption~\ref{assp:symmetric-costs}.
The uniqueness condition of Theorem~\ref{thm:symmetric-uniqueness}
then depends only on $\thetaC$, $\lambda$, and $\rho$, which are
not affected by the subsidy.
At the \CfOffsetN{} parameter vectors where $s > 0$, the condition
holds in \CfOffsetThmShare{} of cases, and the search finds a unique
equilibrium in \CfOffsetUniqueShare{}.
Table~\ref{tab:counterfactual} reports the spreads for the rest.

\section{Conclusion}
\label{sec:conclusion}

We derive closed-form sufficient conditions for equilibrium uniqueness
in continuous-time dynamic discrete games.
The first condition applies to games with any finite number of
players, any finite state space, and binary actions.
It guarantees uniqueness when the discount rate is sufficiently large
relative to move arrival rates and variation in payoffs.
The second condition ensures equilibrium uniqueness in a duopoly model
without sunk costs.
When competition is moderate relative to the frequency of moves,
uniqueness holds even for arbitrarily patient firms.

We then conduct a large-scale computational study of the determinants
of uniqueness and multiplicity in a duopoly model of entry and exit.
Multiplicity occurs primarily in economically relevant regions
involving patient firms, markets where duopoly is unprofitable,
and intermediate move arrival rates.
Elsewhere, we find uniqueness across most of the sampled parameter
space, well beyond the regions covered by our two conditions.
Extending the conditions to cover this larger region is a natural
direction for future work.

Where multiplicity does occur, we investigate its economic and
econometric consequences.
Even at the same parameter value, multiple equilibria may differ
meaningfully in terms of long-run market structure or producer
surplus.
The na{\"i}ve full-solution estimator, which assumes uniqueness, can
track an unstable equilibrium branch even if the data-generating
equilibrium is stable.
The CTNPL estimator, which estimates the equilibrium from the data,
remains reliable when multiplicity is present, but we recommend using
the iterative version only when the CTNPL operator is contractive,
which can be checked at the one-step estimates.
For counterfactuals, predictions of aggregate quantities coincide
across stable counterfactual equilibria under firm symmetry, while
predictions of
quantities involving firm identity largely depend on equilibrium
selection.
Because our theoretical conditions can be checked directly from the
model primitives, they aid researchers in confirming uniqueness rather
than imposing it.

\appendix

\section{Proofs}
\label{sec:proofs}

This section contains the proofs of all formal results stated in the
main text. Following these, we present two supporting lemmas.

\begin{proof}[Proof of Lemma~\ref{lem:value-bounds}]
Write $R = R(V)$ for brevity.
For any $V \in \mathbb{R}^{NK}$, multiplying
\eqref{eq:bellman-general} by $\rho + \eta$ gives, for player $i$ in
state $k$,
$(\rho + \eta) [T(V)]_{ik} = u_{ik} + \sum_{m \neq i} \lambda_m [\Sigma_m(V_m)V_i]_k + \lambda_i E_{ik}(V_i)$.
For any two states $k$ and $l$, subtracting yields:
\begin{multline*}
(\rho + \eta)\left\lbrace [T(V)]_{ik} - [T(V)]_{il} \right\rbrace = (u_{ik} - u_{il}) \\ + \sum_{m \neq i} \lambda_m \left\lbrace [\Sigma_m(V_m) V_i]_k - [\Sigma_m(V_m) V_i]_{l} \right\rbrace + \lambda_i \left[ E_{ik}(V_i) - E_{il}(V_i) \right].
\end{multline*}
Stochastic matrices are nonexpansive under $\spn$.\footnote{
  If $P$ is a stochastic matrix and $z$ is a conformable vector,
  then each entry of $Pz$ is a convex combination of the components of $z$.
  It follows that $\max(Pz) \leq \max z$ and $\min(Pz) \geq \min z$.
  Subtracting yields $\spn(Pz) \leq \spn(z)$.
  Since each entry of $Pz$ lies in $[\min z, \max z]$, we also
  have $\|Pz\|_\infty \leq \|z\|_\infty$.}
Therefore, each transition term satisfies
$\spn(\Sigma_m(V_m) V_i) \leq \spn(V_i) \leq R$.
For the surplus term, note that
$E_{ik}(V_i) = L(\psi_{ik} + V_{i,\ell(i,k)},\, V_{ik})$,
where $\ell(i,k)$ denotes the
destination state when player $i$ switches from state $k$.
By Lemma~\ref{lem:exchangeable}(\ref{lem:ex-lipschitz}), $L$ is
Lipschitz continuous with constant 1 in the $\ell^\infty$ norm.
Comparing states $k$ and $l$, the second arguments of $L$ differ by
at most $\spn(V_i) \leq R$, while the first arguments differ by at most
$\spn(\psi_i) + \spn(V_i) \leq \Psi + R$.
Therefore,
$|E_{ik}(V_i) - E_{il}(V_i)| \leq \max\{R, \Psi + R\} = \Psi + R$.
Combining results gives
\begin{equation*}
(\rho + \eta)\left|[T(V)]_{ik} - [T(V)]_{il}\right| \leq U + \sum_{m \neq i} \lambda_m R + \lambda_i (\Psi + R) = U + \eta R + \lambda_i \Psi.
\end{equation*}
Taking the maximum over players $i$ and state pairs $(k,l)$ gives
\begin{equation*}
(\rho + \eta)\, R(T(V)) \leq U + \eta\, R + \lambar \Psi.
\end{equation*}
Dividing by $\rho + \eta$ gives part~(\ref{lem:vb-operator}).

For part~(\ref{lem:vb-invariant}), if $R(V) \leq M$ then
$(\rho + \eta)\, R(T(V)) \leq U + \eta M + \lambar \Psi
= (\rho + \eta) M - \rho M + U + \lambar \Psi = (\rho + \eta) M$,
where the last equality uses $\rho M = U + \lambar \Psi$
(from the definition of $M$).
Hence $R(T(V)) \leq M$.

For part~(\ref{lem:vb-equilibrium}), at an equilibrium $V^* = T(V^*)$,
so $R(V^*) = R(T(V^*))$.  Substituting into~\eqref{eq:operator-bound}
and rearranging: $\rho\, R(V^*)/(\rho + \eta) \leq (U + \lambar \Psi)/(\rho + \eta)$,
giving $R(V^*) \leq M$.
\end{proof}

\begin{proof}[Proof of Theorem~\ref{thm:contraction}]
We show that $T$ is a contraction on the invariant set $\VM$ defined
in Lemma~\ref{lem:value-bounds}.
Let $V, V' \in \VM$.
From~\eqref{eq:bellman-general}, the operator
difference for player $i$ is
\begin{equation*}
T_i(V) - T_i(V') = \frac{1}{\rho + \eta}\left\lbrace\sum_{m \neq i} \lambda_m \left[ \Sigma_m(V_m) V_i - \Sigma_m(V'_m) V'_i\right] + \lambda_i \left[ E_i(V_i) - E_i(V'_i) \right]\right\rbrace.
\end{equation*}
We bound the two types of terms on the right separately.

\emph{Rivals' transition terms.}
For each rival $m \neq i$, we add and subtract $\Sigma_m(V_m) V'_i$:
\begin{equation*}
\Sigma_m(V_m) V_i - \Sigma_m(V'_m) V'_i = \Sigma_m(V_m)(V_i - V'_i) + (\Sigma_m(V_m) - \Sigma_m(V'_m))V'_i.
\end{equation*}
Since $\Sigma_m(V_m)$ is stochastic (row sums equal 1), the first part satisfies
\begin{equation*}
\|\Sigma_m(V_m)(V_i - V'_i)\|_\infty \leq \|V_i - V'_i\|_\infty.
\end{equation*}
For the second part, let
$\Delta\Sigma_m \equiv \Sigma_m(V_m) - \Sigma_m(V'_m)$.
From~\eqref{eq:Sigma-general}, row $k$ of $\Delta\Sigma_m$ is
$[\sigma_{mk}(V_m) - \sigma_{mk}(V'_m)]\,(\iota_{\ell(m,k)} - \iota_k)^\top$,
where $\iota_k$ denotes the $k$th standard basis vector in
$\mathbb{R}^K$, so
\begin{equation*}
(\Delta\Sigma_m V'_i)_k = [\sigma_{mk}(V_m) - \sigma_{mk}(V'_m)]\,(V'_{i,\ell(m,k)} - V'_{ik}).
\end{equation*}
Each switching probability has the form
$\sigma_{mk} = F(\psi_{mk} + V_{m,\ell(m,k)} - V_{mk})$,
where the argument depends on $V_{m}$ through a difference of two
components.  The change in this argument is bounded by
$|(V_{m,\ell(m,k)} - V'_{m,\ell(m,k)}) - (V_{mk} - V'_{mk})|
\leq 2\|V_{m} - V'_{m}\|_\infty$,
and composing with the Lipschitz constant $\cF$ of $F$
(Lemma~\ref{lem:exchangeable}(\ref{lem:ex-cdf-lip})) gives
\begin{equation*}
|\sigma_{mk}(V_m) - \sigma_{mk}(V'_m)| \leq
2 \cF \|V_{m} - V'_{m}\|_\infty.
\end{equation*}
Since $V' \in \VM$, we also have
$|V'_{i,\ell(m,k)} - V'_{ik}| \leq \spn(V'_i) \leq R(V') \leq M$, giving
\begin{equation*}
\|\Delta\Sigma_m V'_i\|_\infty
  \leq 2 \cF \|V_{m} - V'_{m}\|_\infty \cdot M.
\end{equation*}

\emph{Own surplus term.}
Write $E_{ik}(V_i) = L(\psi_{ik} + V_{i,\ell(i,k)},\, V_{ik})$ as in the
proof of Lemma~\ref{lem:value-bounds}.
Since $V_i$ and $V'_i$ are evaluated at the same state $k$, the switching
cost $\psi_{ik}$ cancels and the arguments of $L$ differ by
$(V_{i,\ell(i,k)} - V'_{i,\ell(i,k)},\, V_{ik} - V'_{ik})$.
By Lemma~\ref{lem:exchangeable}(\ref{lem:ex-lipschitz}),
\begin{equation*}
|E_{ik}(V_i) - E_{ik}(V'_i)| \leq \max\{|V_{i,\ell(i,k)} - V'_{i,\ell(i,k)}|, |V_{ik} - V'_{ik}|\} \leq \|V_i - V'_i\|_\infty.
\end{equation*}

\emph{Combining.}
Collecting the terms above:
\begin{equation*}
\|T_i(V) - T_i(V')\|_\infty
\leq \frac{\lambda_{-i} + \lambda_i}{\rho + \eta}\|V_i - V'_i\|_\infty + \frac{2 \cF M}{\rho + \eta}\sum_{m \neq i} \lambda_m \|V_{m} - V'_{m}\|_\infty,
\end{equation*}
recalling that $\lambda_{-i} = \sum_{m \neq i} \lambda_m$.
Since $\|V_i - V'_i\|_\infty \leq \|V - V'\|_\infty$ and
$\|V_m - V'_m\|_\infty \leq \|V - V'\|_\infty$ for all $m$, while
$\lambda_{-i} + \lambda_i = \eta$ and
$\lambda_{-i} \leq \lambarminus$, taking the maximum over players $i$
gives, for the full $NK$-dimensional operator,
\begin{equation*}
\|T(V) - T(V')\|_\infty
\leq \underbrace{\frac{\eta + 2 \cF \lambarminus M}{\rho + \eta}}_{\gamma}\|V - V'\|_\infty.
\end{equation*}

\emph{Deriving the condition.}
The contraction condition $\gamma < 1$ requires
$\eta + 2 \cF \lambarminus M < \rho + \eta$,
which simplifies to $2 \cF \lambarminus M < \rho$.
Substituting $M = (U + \lambar \Psi)/\rho$ from
Lemma~\ref{lem:value-bounds}:
\begin{equation*}
\frac{2 \cF \lambarminus (U + \lambar \Psi)}{\rho} < \rho,
\end{equation*}
which rearranges to~\eqref{eq:contraction-bound}.

Under this condition, $T$ is a contraction on $\VM$ with respect to
$\|\cdot\|_\infty$.
The set $\VM$ is nonempty (it contains $V = 0$, since $M \geq 0$) and
closed in $\mathbb{R}^{NK}$ (since $R$ is continuous), hence complete.
By Lemma~\ref{lem:value-bounds}(\ref{lem:vb-invariant}),
$T$ maps $\VM$ to itself.
The Banach fixed-point theorem therefore gives a unique fixed
point in $\VM$.  By
Lemma~\ref{lem:value-bounds}(\ref{lem:vb-equilibrium}), every
equilibrium lies in $\VM$, so uniqueness is global.
\end{proof}

\begin{proof}[Proof of Lemma~\ref{lem:symmetric-identities}]
Part (\ref{lem:si-probs}):
$F(-\kappa + d) + F(\kappa - d)
= F(-\kappa + d) + F(-(-\kappa + d)) = 1$,
by Lemma~\ref{lem:exchangeable}(\ref{lem:ex-symmetry}).

\noindent
Part (\ref{lem:si-surplus}):
$e(\thetaSV - d) - e(\thetaEC + d)
= e(\kappa - d) - e(-(\kappa - d)) = \kappa - d$
by Lemma~\ref{lem:exchangeable}(\ref{lem:ex-identity}).

\noindent
Part (\ref{lem:si-rows}):
Consider $\Sigma_2(V_2)$ from~\eqref{eq:Sigma_2}, which is block diagonal
with $2 \times 2$ blocks for state pairs $(1,2)$ and $(3,4)$.
Within each block, the entry probability is
$p = F(\thetaEC + d)$ for the relevant value difference~$d$, and
the exit probability is $F(\thetaSV - d)$.
By part~(\ref{lem:si-probs}),
the exit probability equals $1 - p$, so each block takes the form
$\bigl[\begin{smallmatrix} 1-p & p \\ 1-p & p \end{smallmatrix}\bigr]$
(rank~1, identical rows).  The product $\Sigma_2 V_1$ therefore
yields the same value for states 1 and 2, and for states 3 and 4.
Upon inspecting~\eqref{eq:Sigma_1}, the argument for $\Sigma_1(V_1)$
is analogous: rows 1 and 3 are identical, as are rows 2 and 4.
\end{proof}

\begin{proof}[Proof of Lemma~\ref{lem:constant-delta}]
Define the following value differences for firm~1:
\begin{alignat*}{3}
v_{11} &= V_{13} - V_{11} &\quad& \text{(entry value, rival out),} \\
v_{12} &= V_{14} - V_{12} &\quad& \text{(entry value, rival in),} \\
v_{13} &= V_{12} - V_{11} &\quad& \text{(rival's entry effect, firm~1 out),} \\
v_{14} &= V_{14} - V_{13} &\quad& \text{(rival's entry effect, firm~1 in),}
\end{alignat*}
and analogously for firm~2: $v_{21} = V_{22} - V_{21}$,
$v_{22} = V_{24} - V_{23}$,
$v_{23} = V_{23} - V_{21}$,
$v_{24} = V_{24} - V_{22}$.
Note the identity $v_{i4} = (v_{i2} - v_{i1}) + v_{i3}$.
In this notation,~\eqref{eq:constant-delta} states that
$v_{12} - v_{11}$ and $v_{22} - v_{21}$ both equal
$\thetaC/(\rho + \eta)$.
We show that $v_{i2} - v_{i1}$ is the same constant for each
firm~$i$ in any equilibrium, without assuming equal arrival rates
or a symmetric equilibrium.
Focusing on $x_3 = (1,0)$,
where firm~1's switching action is exit ($x_3 \to x_1$) and firm~2's is entry ($x_3 \to x_4$),
the instantaneous Bellman equation for firm~1
(given in~\eqref{eq:rho-13} in Appendix~\ref{sec:bellman-instantaneous})
can be written:
\begin{align}
\rho V_{13}
            &= \thetaM + \lambda_2 F(\underbrace{\thetaEC}_{-\kappa} + \underbrace{V_{24} - V_{23}}_{v_{22}})(\underbrace{V_{14} - V_{13}}_{v_{14}})
            + \lambda_1 e(\underbrace{\thetaSV}_{\kappa} + \underbrace{V_{11} - V_{13}}_{-v_{11}}) \label{eq:rhoV13}\\
            &= \thetaM + \lambda_2 p_1 v_{14} + \lambda_1 e(\kappa - v_{11}), \notag
\end{align}
where we define $p_1 = F(\thetaEC + v_{22})$ as firm~2's entry probability
in state $x_3$.
Similarly, for $x_1 = (0,0)$, we have
\begin{align*}
\rho V_{11} &= \lambda_2 p_0 v_{13} + \lambda_1 e(-\kappa + v_{11}),
\end{align*}
where $p_0 = F(\thetaEC + v_{21})$ is firm~2's entry probability in
state $x_1$.
Subtracting and applying the surplus identity
of Lemma~\ref{lem:symmetric-identities}(\ref{lem:si-surplus}),
$e(\kappa - v_{11}) - e(-\kappa + v_{11}) = \kappa - v_{11}$,
gives
\begin{equation}
\label{eq:v11}
(\rho + \lambda_1)\, v_{11}
= \thetaM + \lambda_1\kappa
  + \lambda_2 p_1\, v_{14} - \lambda_2 p_0\, v_{13}.
\end{equation}
Similarly, using parts~(\ref{lem:si-probs})
and~(\ref{lem:si-surplus}) of
Lemma~\ref{lem:symmetric-identities} again, subtracting the
Bellman equations for states
$x_4 = (1,1)$ and $x_2 = (0,1)$ gives:
\begin{equation}
\label{eq:v12}
(\rho + \lambda_1)\, v_{12}
= \thetaM + \thetaC + \lambda_1\kappa
  - \lambda_2(1-p_1)\, v_{14} + \lambda_2(1-p_0)\, v_{13}.
\end{equation}
Subtracting~\eqref{eq:v11} from~\eqref{eq:v12},
the $p_0$ and $p_1$ terms collapse:
$-\lambda_2(1-p_1)v_{14} - \lambda_2 p_1 v_{14} = -\lambda_2 v_{14}$ and
$\lambda_2(1-p_0)v_{13} + \lambda_2 p_0 v_{13} = \lambda_2 v_{13}$.
The result is:
\begin{equation*}
(\rho + \lambda_1)(v_{12} - v_{11}) = \thetaC - \lambda_2 v_{14} + \lambda_2 v_{13}.
\end{equation*}
Substituting the identity $v_{14} = (v_{12} - v_{11}) + v_{13}$, the $v_{13}$
terms cancel:
\begin{equation*}
(\rho + \eta)(v_{12} - v_{11}) = \thetaC.
\end{equation*}
The same argument applies to firm~2, swapping $x_2$ with $x_3$ and
$\lambda_1$ with $\lambda_2$, and yields the same denominator
$\rho + \eta$.  Defining
$\Delta = \thetaC/(\rho + \eta)$, we have
$v_{i2} = v_{i1} + \Delta$ for $i = 1, 2$, which
is~\eqref{eq:constant-delta}.
\end{proof}

\begin{proof}[Proof of Theorem~\ref{thm:symmetric-uniqueness}]
The proof proceeds in two steps, using the value differences
$v_{ij}$ and the constant $\Delta$ defined in the proof of
Lemma~\ref{lem:constant-delta}.
We first use Lemma~\ref{lem:constant-delta} to reduce the
eight-dimensional system to a two-dimensional fixed-point problem.
We then show that the resulting map is a contraction.

\medskip
\noindent
\textit{Step 1: Reduction to two dimensions.}
Lemma~\ref{lem:constant-delta} reduces the four unknowns per
firm to two: $v_{i1}$ and $v_{i3}$
(since $v_{i2} = v_{i1} + \Delta$ and $v_{i4} = \Delta + v_{i3}$).
We now show that $v_{i3}$ depends only on $v_{i1}$, further reducing
to a single unknown per firm.
The ``firm~1 out'' states are $x_1 = (0,0)$ and $x_2 = (0,1)$.
From~\eqref{eq:bellman-firm1}, their Bellman equations are
\begin{equation*}
(\rho + \eta) V_{11} = \lambda_2[\Sigma_2(V_2)\, V_1]_1 + \lambda_1 E_{11}
\quad\text{and}\quad
(\rho + \eta) V_{12} = \lambda_2[\Sigma_2(V_2)\, V_1]_2 + \lambda_1 E_{12},
\end{equation*}
where the flow payoffs are $u_{11} = u_{12} = 0$ since firm~1 is
inactive in both states.
By~\eqref{eq:sigma2-rows},
$[\Sigma_2(V_2)\, V_1]_1 = [\Sigma_2(V_2)\, V_1]_2$,
so subtracting gives
\begin{equation}
  \label{eq:rho-eta-v13}
  (\rho + \eta) v_{13} = \lambda_1(E_{12} - E_{11}).
\end{equation}
For the social surplus terms, the representation
$E_{ik}(V_i) = V_{ik} + e(\psi_{ik} + V_{i,\ell(i,k)} - V_{ik})$
gives
\begin{equation*}
E_{11} = V_{11} + e(-\kappa + v_{11})
\quad\text{and}\quad
E_{12} = V_{12} + e(-\kappa + v_{12}),
\end{equation*}
so $E_{12} - E_{11}
= v_{13} + e(-\kappa + v_{12}) - e(-\kappa + v_{11})
= v_{13} + e(-\kappa + v_{11} + \Delta) - e(-\kappa + v_{11})$,
where the second step uses $v_{12} = v_{11} + \Delta$ from
Lemma~\ref{lem:constant-delta}.
Substituting this into~\eqref{eq:rho-eta-v13}, recalling that
$\eta = \lambda_1 + \lambda_2$, and collecting $v_{13}$ terms yields:
\begin{equation*}
v_{13} = \phi_1(v_{11}) \equiv
\frac{\lambda_1[e(-\kappa + v_{11} + \Delta) - e(-\kappa + v_{11})]}
     {\rho + \lambda_2}.
\end{equation*}
Substituting $v_{14} = \Delta + v_{13}$ and
$v_{13} = \phi_1(v_{11})$ into~\eqref{eq:v11}:
\begin{align*}
(\rho + \lambda_1)\, v_{11}
&= \thetaM + \lambda_1\kappa
  + \lambda_2 p_1(\Delta + \phi_1(v_{11}))
  - \lambda_2 p_0\, \phi_1(v_{11}) \\
&= \thetaM + \lambda_1\kappa
  + \lambda_2 p_1 \Delta
  + \lambda_2(p_1 - p_0)\, \phi_1(v_{11}).
\end{align*}
For brevity, we now write $\tilde{v}_i = v_{i1}$.
The entry probabilities can be expressed in terms of $\tilde{v}_2$
using $v_{22} = \tilde{v}_2 + \Delta$:
$p_0(\tilde{v}_2) = F(-\kappa + \tilde{v}_2)$ and
$p_1(\tilde{v}_2) = F(-\kappa + \tilde{v}_2 + \Delta)$.
Substituting these in yields a fixed-point equation
$\tilde{v}_1 = G_1(\tilde{v}_1, \tilde{v}_2)$ with
\begin{equation}
\label{eq:G1}
G_1(\tilde{v}_1, \tilde{v}_2)
= \frac{1}{\rho + \lambda_1}[\thetaM + \lambda_1\kappa
  + \lambda_2 p_1(\tilde{v}_2) \Delta
  + \lambda_2(p_1(\tilde{v}_2) - p_0(\tilde{v}_2))\, \phi_1(\tilde{v}_1)].
\end{equation}
The analogous reduction for firm~2---swapping $x_2$ with $x_3$ and
$\lambda_1$ with $\lambda_2$, and
using~\eqref{eq:sigma1-rows} in place
of~\eqref{eq:sigma2-rows}---yields $\tilde{v}_2 = G_2(\tilde{v}_1, \tilde{v}_2)$.
Under the equal-rate assumption $\lambda_1 = \lambda_2 = \lambda$,
$\phi_1 = \phi_2 \equiv \phi$ and
$G_2(\tilde{v}_1, \tilde{v}_2) = G_1(\tilde{v}_2, \tilde{v}_1)$.

\medskip
\noindent
\textit{Step 2: Contraction.}
Since $f$ is continuous, $G$ is $C^1$ on the convex domain $\mathbb{R}^2$.
Therefore, if $\|J_G\|_\infty$, the maximum absolute row sum, is uniformly
bounded above by a constant
$\gamma < 1$, then the mean-value inequality
\citep[3.2.3]{ortega-rheinboldt-1970} implies that $G$ is a
contraction with modulus $\gamma$ in the $\ell^\infty$ norm.
The Jacobian of $G = (G_1, G_2)$ has the form
$J_G = \bigl[\begin{smallmatrix} a & b \\ b' & a'
\end{smallmatrix}\bigr]$
where $a = \partial G_1/\partial \tilde{v}_1$,
$b = \partial G_1/\partial \tilde{v}_2$,
and the symmetry $G_1(\tilde{v}_1, \tilde{v}_2) = G_2(\tilde{v}_2, \tilde{v}_1)$ gives
$a' = a(\tilde{v}_2, \tilde{v}_1)$ and $b' = b(\tilde{v}_2, \tilde{v}_1)$.
Therefore both row sums satisfy the same bound, and it suffices
to bound $|a| + |b|$.

From~\eqref{eq:G1}, $G_1$ depends on $\tilde{v}_1$ only through
$\phi_1(\tilde{v}_1)$ and on $\tilde{v}_2$ only through $p_0(\tilde{v}_2)$ and $p_1(\tilde{v}_2)$.
We will use the following three bounds:
\begin{enumerate}[(i)]
  \item $|p_1(\tilde{v}_2) - p_0(\tilde{v}_2)|
    = |F(-\kappa + \tilde{v}_2 + \Delta) - F(-\kappa + \tilde{v}_2)|
    \leq \cF |\Delta|$
    by Lemma~\ref{lem:exchangeable}(\ref{lem:ex-cdf-lip}),
  \item $|\phi_1(\tilde{v}_1)| \leq \lambda_1|\Delta|/(\rho + \lambda_2)$,
    since $e'(z) = F(z) \in (0,1)$
    by Lemma~\ref{lem:exchangeable}(\ref{lem:ex-deriv}),
  \item $|\phi_1'(\tilde{v}_1)| \leq \cF \lambda_1|\Delta|/(\rho+\lambda_2)$,
\end{enumerate}
where (iii) follows by differentiating $\phi_1$ and
applying Lemma~\ref{lem:exchangeable}(\ref{lem:ex-deriv})
to obtain
$\phi_1'(\tilde{v}_1)
= \frac{\lambda_1[F(-\kappa + \tilde{v}_1 + \Delta)
        - F(-\kappa + \tilde{v}_1)]}{\rho + \lambda_2}$,
and then applying Lemma~\ref{lem:exchangeable}(\ref{lem:ex-cdf-lip}) again.

Next, to bound $a$ we differentiate~\eqref{eq:G1} with respect to
$\tilde{v}_1$ and apply (i) and (iii) to find:
\begin{equation}
\label{eq:a-bound}
a = \frac{\lambda_2(p_1 - p_0)}{\rho + \lambda_1}\,\phi_1'(\tilde{v}_1)
\quad\text{and}\quad
|a| \leq \frac{\cF^2 \lambda_1\lambda_2\,\Delta^2}
     {(\rho+\lambda_1)(\rho+\lambda_2)}.
\end{equation}

Similarly, to bound $b$ we differentiate~\eqref{eq:G1} with respect
to $\tilde{v}_2$, writing $p_0' = dp_0/d\tilde{v}_2$ and $p_1' = dp_1/d\tilde{v}_2$.
By Assumptions~\ref{assp:bounded-density} and~\ref{assp:exchangeable},
$p_0' = f(-\kappa + \tilde{v}_2) \leq \cF$, $p_1' = f(-\kappa + \tilde{v}_2 + \Delta) \leq \cF$,
and $f$ is Lipschitz with constant $\cFp$, so
$|p_1' - p_0'| = |f(-\kappa + \tilde{v}_2 + \Delta) - f(-\kappa + \tilde{v}_2)| \leq \cFp |\Delta|$.
Applying these together with~(ii):
\begin{equation}
\label{eq:b-bound}
b = \frac{\lambda_2}{\rho + \lambda_1} [p_1'\,\Delta + (p_1' - p_0')\,\phi_1]
\quad\text{and}\quad
|b| \leq \frac{\cF \lambda_2|\Delta|}{\rho+\lambda_1}
  + \frac{\cFp \lambda_1\lambda_2\,\Delta^2}
         {(\rho+\lambda_1)(\rho+\lambda_2)}.
\end{equation}

Combining~\eqref{eq:a-bound} and~\eqref{eq:b-bound}
with $\lambda_1 = \lambda_2 = \lambda$:
\begin{align*}
\|J_G\|_\infty \leq |a| + |b|
&\leq \frac{\cF^2 \lambda^2\Delta^2}{(\rho+\lambda)^2}
  + \frac{\cF \lambda|\Delta|}{\rho+\lambda}
  + \frac{\cFp \lambda^2\Delta^2}{(\rho+\lambda)^2} \\
&= \frac{\cF \lambda|\Delta|}{\rho+\lambda}
  + \frac{(\cF^2 + \cFp) \lambda^2\Delta^2}{(\rho+\lambda)^2}.
\end{align*}
Since $\Delta = \thetaC/(\rho + 2\lambda)$, we have
$\lambda\Delta/(\rho+\lambda) = \tilde\Delta$, and therefore
\begin{equation*}
\|J_G\|_\infty \leq \cF|\tilde\Delta| + (\cF^2 + \cFp)\,\tilde\Delta^2.
\end{equation*}
Under condition~\eqref{eq:symmetric-contraction} of the theorem,
$\|J_G\|_\infty < 1$ and so $G$ is a contraction on
$(\mathbb{R}^2, \|\cdot\|_\infty)$.
By the Banach fixed-point theorem \citep[5.1.3]{ortega-rheinboldt-1970},
$G$ has a unique fixed point $(\tilde{v}_1^*, \tilde{v}_2^*)$.

We now show that this uniqueness extends to the full eight-dimensional
system.
Because every equilibrium induces a fixed point of $G$, all
equilibria share the value differences $(\tilde{v}_1^*, \tilde{v}_2^*)$ and hence,
by Lemma~\ref{lem:constant-delta} and Step~1, all of the
differences $v_{ij}$.
The levels are determined as well.
In the instantaneous Bellman equation~\eqref{eq:rhoV13}
(see also Appendix~\ref{sec:bellman-instantaneous}),
the right-hand side depends only on value differences, so with $\rho > 0$
each value $V_{ik}$ is uniquely determined.
Hence there is at most one equilibrium.
Since equilibrium existence is guaranteed by Theorem~1
of \cite{ctgames}, it follows that there is exactly one
equilibrium.

It remains to show that the equilibrium is symmetric.
Let $\Xi \colon \mathbb{R}^8 \to \mathbb{R}^8$ denote the map that swaps the firm
labels together with the states $x_2$ and $x_3$:
\begin{equation*}
\Xi V \equiv (V_{21}, V_{23}, V_{22}, V_{24}, V_{11}, V_{13}, V_{12}, V_{14}).
\end{equation*}
With $\lambda_1 = \lambda_2$, the Bellman system commutes with this
relabeling, $T(\Xi V) = \Xi\,T(V)$: for each state, the Bellman
equation for firm~1 becomes the equation for firm~2 after swapping
$x_2$ with $x_3$ and $\lambda_1$ with $\lambda_2$
(cf.~\eqref{eq:rho-13} and~\eqref{eq:rho-22}
in Appendix~\ref{sec:bellman-instantaneous}).
Therefore $\Xi V^*$ is an equilibrium whenever $V^*$ is, and uniqueness implies
$\Xi V^* = V^*$: $V_{11} = V_{21}$, $V_{12} = V_{23}$, $V_{13} = V_{22}$, and
$V_{14} = V_{24}$.
\end{proof}

\begin{proof}[Proof of Corollary~\ref{cor:discount-free}]
Let $g(\rho)$ denote the left-hand side of~\eqref{eq:symmetric-contraction},
holding $\thetaC$ and $\lambda$ fixed.
Note that $g$ is non-increasing in $\rho$, so
\begin{equation*}
\sup_{\rho>0} g(\rho) \leq \lim_{\rho\to 0^+} g(\rho)
= \frac{\cF |\thetaC|}{2\lambda} + \frac{(\cF^2 + \cFp) \thetaC^2}{4\lambda^2} < 1,
\end{equation*}
where the last inequality is the condition of the corollary.
Therefore~\eqref{eq:symmetric-contraction} holds for every $\rho > 0$
and the conclusion follows from Theorem~\ref{thm:symmetric-uniqueness}.
\end{proof}

\subsection*{Supporting Lemmas}

Parts (\ref{lem:ex-deriv}) and (\ref{lem:ex-gradient}) of the
following lemma are related to the two-alternative case of the
Williams--Daly--Zachary theorem: the gradient of the surplus
function with respect to the deterministic utility components is
the vector of choice probabilities
\citep[Section~5.8, Theorem~5.1]{mcfadden81econometric};
see also \citet[Theorem~3.1]{rust-1994-structural}.

\begin{lemma}[Error Distribution Properties]
\label{lem:exchangeable}
Under Assumptions~\ref{assp:pi}, \ref{assp:bounded-density},
and~\ref{assp:exchangeable}:
\begin{enumerate}[(a)]
\item \label{lem:ex-deriv} $e'(z) = F(z) \in (0, 1)$ for all
  $z \in \mathbb{R}$.
\item \label{lem:ex-gradient} For all $z_1, z_2, c \in \mathbb{R}$,
  $L(z_1 + c, z_2 + c) = L(z_1, z_2) + c$, so that
  $L(z_1, z_2) = z_2 + e(z_1 - z_2)$, with
  $\frac{\partial L}{\partial z_1} = F(z_1 - z_2)$ and
  $\frac{\partial L}{\partial z_2} = 1 - F(z_1 - z_2)$.
\item \label{lem:ex-lipschitz} $L$ is Lipschitz continuous with constant
  1 in the $\ell^\infty$ norm:
  $|L(z_1, z_2) - L(z_1', z_2')| \leq \max\{|z_1 - z_1'|, |z_2 - z_2'|\}$
  for all $z_1, z_2, z_1', z_2' \in \mathbb{R}$.
\item \label{lem:ex-cdf-lip} $\abs{F(z) - F(z')} \leq \cF \abs{z - z'}$
  for all $z, z' \in \mathbb{R}$.
\item \label{lem:ex-symmetry} $F(z) + F(-z) = 1$ for all
  $z \in \mathbb{R}$.
\item \label{lem:ex-identity} $e(z) - e(-z) = z$ for all
  $z \in \mathbb{R}$.
\end{enumerate}
Specifically,
(\ref{lem:ex-deriv}), (\ref{lem:ex-gradient}), and (\ref{lem:ex-lipschitz}) require only Assumption~\ref{assp:pi};
(\ref{lem:ex-cdf-lip}) requires Assumptions~\ref{assp:pi} and~\ref{assp:bounded-density};
(\ref{lem:ex-symmetry}) and (\ref{lem:ex-identity}) require Assumptions~\ref{assp:pi} and~\ref{assp:exchangeable}.
\end{lemma}

\begin{proof}
The properties are independent of the player and state, so for
brevity write
$\varepsilon_0 = \varepsilon_{i0k}$ and
$\varepsilon_1 = \varepsilon_{i1k}$, so that
$e(z) = \E_{\varepsilon} \max\{z + \varepsilon_1, \varepsilon_0\}$.

Part (\ref{lem:ex-deriv}):
For each realization of $(\varepsilon_0, \varepsilon_1)$, the
integrand $\max\{z + \varepsilon_1, \varepsilon_0\}$ in the
definition of $e$ is 1-Lipschitz in $z$.
The derivative is $\1\{z + \varepsilon_1 > \varepsilon_0\}$ wherever
$z + \varepsilon_1 \neq \varepsilon_0$, an event of probability one.
The difference quotients are therefore bounded by one and converge
almost surely, so dominated convergence
\citep[Theorem~1.6.7]{durrett-2019} gives
$e'(z) = \Pr[\varepsilon_0 - \varepsilon_1 < z]
= \Pr[\varepsilon_0 - \varepsilon_1 \leq z] = F(z)$,
where the second equality holds because the difference is
continuously distributed.
The value $F(z)$ lies in $(0, 1)$ because the difference has full support.

Part (\ref{lem:ex-gradient}):
Since $\max\{z_1 + c + \varepsilon_1, z_2 + c + \varepsilon_0\}
= \max\{z_1 + \varepsilon_1, z_2 + \varepsilon_0\} + c$,
taking expectations gives $L(z_1 + c, z_2 + c) = L(z_1, z_2) + c$.
Setting $c = -z_2$ gives
$L(z_1, z_2) = z_2 + L(z_1 - z_2, 0) = z_2 + e(z_1 - z_2)$.
Differentiating and applying part~(\ref{lem:ex-deriv}) yields
$\partial L/\partial z_1 = e'(z_1 - z_2) = F(z_1 - z_2)$ and
$\partial L/\partial z_2 = 1 - e'(z_1 - z_2) = 1 - F(z_1 - z_2)$.

Part (\ref{lem:ex-lipschitz}):
For each realization of $(\varepsilon_0, \varepsilon_1)$,
$\abs{\max\{z_1 + \varepsilon_1, z_2 + \varepsilon_0\}
- \max\{z_1' + \varepsilon_1, z_2' + \varepsilon_0\}}
\leq \max\{\abs{z_1 - z_1'}, \abs{z_2 - z_2'}\}$.
Taking expectations preserves the inequality.

Part (\ref{lem:ex-cdf-lip}):
$\abs{F(z) - F(z')} = \abs{\int_{z'}^{z} f(s) \, ds}
\leq \cF \abs{z - z'}$.

Part (\ref{lem:ex-symmetry}):
By exchangeability, $\varepsilon_0 - \varepsilon_1$ and
$\varepsilon_1 - \varepsilon_0$ have the same distribution, so the
difference is symmetric about zero.
Since the difference is continuously distributed,
$F(-z) = \Pr[\varepsilon_0 - \varepsilon_1 \leq -z]
= \Pr[\varepsilon_0 - \varepsilon_1 \geq z] = 1 - F(z)$.

Part (\ref{lem:ex-identity}):
Using exchangeability and
$\max\{a, b\} - z = \max\{a - z, b - z\}$, we have
$e(-z) = \E_{\varepsilon} \max\{\varepsilon_1 - z, \varepsilon_0\}
= \E_{\varepsilon} \max\{\varepsilon_0 - z, \varepsilon_1\}
= \E_{\varepsilon} \max\{\varepsilon_0, \varepsilon_1 + z\} - z = e(z) - z$.
\end{proof}

\begin{lemma}[Logit Constants]
\label{lem:logit-lipschitz}
If Assumption~\ref{assp:ev1} holds,
then $F = \Lambda$, $L(z_1,z_2) = \log[\exp(z_1) + \exp(z_2)]$, and
$e(z) = L(z, 0) = \log[1 + \exp(z)]$.
Furthermore, Assumptions~\ref{assp:bounded-density} and~\ref{assp:exchangeable}
are satisfied with
$\cF = \sup_z \Lambda'(z) = \cFLogit{}$ and
$\cFp = \sup_z \abs{\Lambda''(z)} = \cFpLogit{}$.
\end{lemma}

\begin{proof}
The errors are independent across choices with a common marginal
distribution, so the pair
$(\varepsilon_{i0k}, \varepsilon_{i1k})$ is exchangeable.
The difference of two independent type I extreme value random
variables follows the standard logistic distribution, so $F = \Lambda$,
and for independent mean-zero type I extreme value errors, the surplus
function is $L(z_1,z_2) = \log[\exp(z_1)+\exp(z_2)]$.
For the constants, write $t = \Lambda(z) \in (0, 1)$, so that
$\Lambda'(z) = t(1 - t)$, which is maximized at $t = 1/2$, giving
$\cF = \cFLogit{}$.
Differentiating again gives
$\Lambda''(z) = \Lambda'(z)(1 - 2\Lambda(z)) = t(1-t)(1-2t)$.
The derivative with respect to $t$ is $6t^2 - 6t + 1$, which has
roots $t = (3 \pm \sqrt{3})/6$
at which $\abs{\Lambda''(z)} = \cFpLogit{}$.
This is the Lipschitz constant since $\Lambda''$ is continuous and bounded.
\end{proof}

\section{Instantaneous Bellman Equations}
\label{sec:bellman-instantaneous}

The Bellman equations~\eqref{eq:bellman-firm1}
and~\eqref{eq:bellman-firm2} can be restated in instantaneous
($\rho$-discounted) form by multiplying by
$\rho + \eta$ and subtracting
$\eta V_{ik} = (\lambda_1 + \lambda_2) V_{ik}$ from each equation.
Substituting the choice probabilities
$\sigma_{ik} = F(\psi_{ik} + V_{i,\ell(i,k)} - V_{ik})$
from~\eqref{eq:ccp} and the surplus
$E_{ik}(V_i) = V_{ik} + e(\psi_{ik} + V_{i,\ell(i,k)} - V_{ik})$
yields the following system:
\begin{align}
\rho V_{11} &= \lambda_2 F(\thetaEC + V_{22} - V_{21})
  (V_{12} - V_{11})
  + \lambda_1 e(\thetaEC + V_{13} - V_{11})
  \label{eq:rho-11} \\
\rho V_{12} &= \lambda_2 F(\thetaSV + V_{21} - V_{22})
  (V_{11} - V_{12})
  + \lambda_1 e(\thetaEC + V_{14} - V_{12})
  \label{eq:rho-12} \\
\rho V_{13} &= \thetaM
  + \lambda_2 F(\thetaEC + V_{24} - V_{23})
  (V_{14} - V_{13})
  + \lambda_1 e(\thetaSV + V_{11} - V_{13})
  \label{eq:rho-13} \\
\rho V_{14} &= \thetaM + \thetaC
  + \lambda_2 F(\thetaSV + V_{23} - V_{24})
  (V_{13} - V_{14})
  + \lambda_1 e(\thetaSV + V_{12} - V_{14})
  \label{eq:rho-14} \\
\rho V_{21} &= \lambda_1 F(\thetaEC + V_{13} - V_{11})
  (V_{23} - V_{21})
  + \lambda_2 e(\thetaEC + V_{22} - V_{21})
  \label{eq:rho-21} \\
\rho V_{22} &= \thetaM
  + \lambda_1 F(\thetaEC + V_{14} - V_{12})
  (V_{24} - V_{22})
  + \lambda_2 e(\thetaSV + V_{21} - V_{22})
  \label{eq:rho-22} \\
\rho V_{23} &= \lambda_1 F(\thetaSV + V_{11} - V_{13})
  (V_{21} - V_{23})
  + \lambda_2 e(\thetaEC + V_{24} - V_{23})
  \label{eq:rho-23} \\
\rho V_{24} &= \thetaM + \thetaC
  + \lambda_1 F(\thetaSV + V_{12} - V_{14})
  (V_{22} - V_{24})
  + \lambda_2 e(\thetaSV + V_{23} - V_{24})
  \label{eq:rho-24}
\end{align}
Each equation has up to three terms: a flow payoff
(zero when the firm is inactive), a rival-transition
term (weighted by $\lambda_{-i}$), and a surplus term
(weighted by $\lambda_i$).  The structural symmetry
between the two firms is apparent: \eqref{eq:rho-13} for firm~1
(monopolist in $x_3 = (1,0)$) has the same structure as
\eqref{eq:rho-22} for firm~2 (monopolist in $x_2 = (0,1)$),
with $\lambda_1$ and $\lambda_2$ interchanged.

\section{Analytical Jacobian}
\label{sec:jacobian}

For computational efficiency, we derive the analytical Jacobian of the
equilibrium system.
Throughout this appendix we impose the logit specification of
Assumption~\ref{assp:ev1}, as maintained in the numerical work,
with $F = \Lambda$ and $e(z) = \log[1 + \exp(z)]$.
The equilibrium conditions $H(V) = V - T(V) = 0$ have Jacobian
$J = I - DT$, where $DT$ is the Jacobian of $T$.
The Jacobian has block structure
\begin{equation*}
DT = \begin{bmatrix} DT_{11} & DT_{12} \\ DT_{21} &
  DT_{22} \end{bmatrix}
\end{equation*}
where $DT_{im} = \frac{\partial T_i}{\partial V_m}$ (each block is $4 \times 4$).

We consider each block in turn.
The block $DT_{11} = \frac{\partial T_1}{\partial V_1}$ contains the
derivatives of firm~1's values with respect to firm~1's own values.
Recalling the definition of $T_1$ in~\eqref{eq:bellman-firm1} and differentiating
gives:
\begin{equation*}
DT_{11} = \frac{1}{\rho + \eta}\left[ \lambda_2 \Sigma_2(V_2) + \lambda_1 \frac{\partial E_1}{\partial V_1} \right].
\end{equation*}
To express this matrix concisely, we write out the switching
probabilities $\sigma_{ik}$ of~\eqref{eq:ccp} for firm $i$ in state
$k$ in full:
\begin{align*}
\sigma_{11} &= \Lambda(\thetaEC + V_{13} - V_{11}), \quad \sigma_{12} = \Lambda(\thetaEC + V_{14} - V_{12}), \\
\sigma_{13} &= \Lambda(\thetaSV + V_{11} - V_{13}), \quad \sigma_{14} = \Lambda(\thetaSV + V_{12} - V_{14}), \\
\sigma_{21} &= \Lambda(\thetaEC + V_{22} - V_{21}), \quad \sigma_{22} = \Lambda(\thetaSV + V_{21} - V_{22}), \\
\sigma_{23} &= \Lambda(\thetaEC + V_{24} - V_{23}), \quad \sigma_{24} = \Lambda(\thetaSV + V_{23} - V_{24}).
\end{align*}
To keep the notation compact, we write $\tilde{\sigma}_{ik} \equiv 1 - \sigma_{ik}$
for the continuation probability of firm $i$ in state $x_k$, so that
$\sigma_{ik} + \tilde{\sigma}_{ik} = 1$.

The first term, $\lambda_2 \Sigma_2(V_2)$, comes from differentiating
the matrix-vector product $\Sigma_2(V_2) V_1$ with respect to $V_1$.
Since $\Sigma_2$ depends only on $V_2$, we have
$\frac{\partial}{\partial V_1}[\Sigma_2(V_2) V_1] = \Sigma_2(V_2)$.
The second term involves derivatives of the social surplus.
For $E_{1k} = \log[\exp(\psi_{1k} + V_{1,\ell(1,k)}) + \exp(V_{1k})]$,
where $\ell(1,k)$ is the destination state when firm~1 switches from state $k$,
we apply Lemma~\ref{lem:exchangeable}(\ref{lem:ex-gradient})
with $F = \Lambda$, $z_1 = \psi_{1k} + V_{1,\ell(1,k)}$, and $z_2 = V_{1k}$:
\begin{align*}
\frac{\partial E_{1k}}{\partial V_{1k}} &= 1 - \Lambda(\psi_{1k} + V_{1,\ell(1,k)} - V_{1k}) = \tilde{\sigma}_{1k}, \\
\frac{\partial E_{1k}}{\partial V_{1,\ell(1,k)}} &= \Lambda(\psi_{1k} + V_{1,\ell(1,k)} - V_{1k}) = \sigma_{1k}, \\
\frac{\partial E_{1k}}{\partial V_{1k'}} &= 0 \quad \text{for } k' \neq k, \ell(1,k).
\end{align*}
Thus each row of $\frac{\partial E_1}{\partial V_1}$ has exactly two
nonzero entries: the continuation probability $\tilde{\sigma}_{1k}$ on
the diagonal and the switching probability $\sigma_{1k}$ at
column~$\ell(1,k)$.
The complete $4 \times 4$ matrix structure is:
\begin{equation*}
\frac{\partial E_1}{\partial V_1} = \begin{bmatrix}
\tilde{\sigma}_{11} & 0 & \sigma_{11} & 0 \\
0 & \tilde{\sigma}_{12} & 0 & \sigma_{12} \\
\sigma_{13} & 0 & \tilde{\sigma}_{13} & 0 \\
0 & \sigma_{14} & 0 & \tilde{\sigma}_{14}
\end{bmatrix}.
\end{equation*}
Note that by Lemma~\ref{lem:exchangeable}(\ref{lem:ex-gradient}), each row sums to 1.
Indeed, comparing with~\eqref{eq:Sigma_1}, the surplus derivative
matrix is the transition matrix itself,
$\partial E_1/\partial V_1 = \Sigma_1(V_1)$, with the $F(\cdot)$
entries of~\eqref{eq:Sigma_1} written here as the switching
probabilities $\sigma_{1k}$.

Combining the two terms, the complete expression for $DT_{11}$ is:
\begin{equation*}
DT_{11} = \frac{1}{\rho + \eta}\left[\begin{smallmatrix}
\lambda_2 \tilde{\sigma}_{21} + \lambda_1 \tilde{\sigma}_{11} & \lambda_2 \sigma_{21} & \lambda_1 \sigma_{11} & 0 \\
\lambda_2 \sigma_{22} & \lambda_2 \tilde{\sigma}_{22} + \lambda_1 \tilde{\sigma}_{12} & 0 & \lambda_1 \sigma_{12} \\
\lambda_1 \sigma_{13} & 0 & \lambda_2 \tilde{\sigma}_{23} + \lambda_1 \tilde{\sigma}_{13} & \lambda_2 \sigma_{23} \\
0 & \lambda_1 \sigma_{14} & \lambda_2 \sigma_{24} & \lambda_2 \tilde{\sigma}_{24} + \lambda_1 \tilde{\sigma}_{14}
\end{smallmatrix}\right].
\end{equation*}

Next, we turn to the block
$DT_{12} = \frac{\partial T_1}{\partial V_2}$, the derivative of firm~1's
values with respect to firm~2's values.
From the definition of $T_1$, we can see that this involves
differentiating $\Sigma_2(V_2) V_1$ with respect to $V_2$.
The $(k,l)$ element of $DT_{12}$ is given by:
\begin{equation*}
(DT_{12})_{kl} = \frac{\lambda_2}{\rho + \eta} \frac{\partial}{\partial V_{2l}}\left[(\Sigma_2 V_1)_k\right]
       = \frac{\lambda_2}{\rho + \eta} \sum_{k'=1}^{4} \frac{\partial (\Sigma_2)_{kk'}}{\partial V_{2l}} V_{1k'}.
\end{equation*}
For each element of $\Sigma_2$ that depends on $V_2$, we use $\Lambda'(z) = \Lambda(z)(1-\Lambda(z))$.
The derivatives of the choice probabilities are:
\begin{align*}
\frac{\partial \sigma_{2k}}{\partial V_{2l}} &= \begin{cases}
\sigma_{2k} \tilde{\sigma}_{2k} & \text{if } V_{2l} \text{ appears with a positive sign in } \sigma_{2k} \\
-\sigma_{2k} \tilde{\sigma}_{2k} & \text{if } V_{2l} \text{ appears with a negative sign in } \sigma_{2k} \\
0 & \text{otherwise}.
\end{cases}
\end{align*}

Since rows 1 and 2 of $\Sigma_2$ depend only on $(V_{21}, V_{22})$ and
rows 3 and 4 depend only on $(V_{23}, V_{24})$, the complete
$4 \times 4$ matrix is:
\begin{equation*}
DT_{12} = \frac{\lambda_2}{\rho + \eta}\left[\begin{smallmatrix}
\sigma_{21} \tilde{\sigma}_{21}(V_{11} - V_{12}) & \sigma_{21} \tilde{\sigma}_{21}(V_{12} - V_{11}) & 0 & 0 \\
\sigma_{22} \tilde{\sigma}_{22}(V_{11} - V_{12}) & \sigma_{22} \tilde{\sigma}_{22}(V_{12} - V_{11}) & 0 & 0 \\
0 & 0 & \sigma_{23} \tilde{\sigma}_{23}(V_{13} - V_{14}) & \sigma_{23} \tilde{\sigma}_{23}(V_{14} - V_{13}) \\
0 & 0 & \sigma_{24} \tilde{\sigma}_{24}(V_{13} - V_{14}) & \sigma_{24} \tilde{\sigma}_{24}(V_{14} - V_{13})
\end{smallmatrix}\right].
\end{equation*}
Each entry is weighted by firm~2's switching probability sensitivity
$\sigma_{2k} \tilde{\sigma}_{2k}$ times the difference in firm~1's values
that appears in row $k$ of $\Sigma_2$.

By symmetry with $DT_{12}$, the block $DT_{21} = \frac{\partial T_2}{\partial V_1}$ has $(k,l)$ element:
\begin{equation*}
(DT_{21})_{kl} = \frac{\lambda_1}{\rho + \eta} \sum_{k'=1}^{4} \frac{\partial (\Sigma_1)_{kk'}}{\partial V_{1l}} V_{2k'}.
\end{equation*}
Since rows 1 and 3 of $\Sigma_1$ depend only on $(V_{11}, V_{13})$ and
rows 2 and 4 depend only on $(V_{12}, V_{14})$, the complete
$4 \times 4$ matrix is:
\begin{equation*}
DT_{21} = \frac{\lambda_1}{\rho + \eta}\left[\begin{smallmatrix}
\sigma_{11} \tilde{\sigma}_{11}(V_{21} - V_{23}) & 0 & \sigma_{11} \tilde{\sigma}_{11}(V_{23} - V_{21}) & 0 \\
0 & \sigma_{12} \tilde{\sigma}_{12}(V_{22} - V_{24}) & 0 & \sigma_{12} \tilde{\sigma}_{12}(V_{24} - V_{22}) \\
\sigma_{13} \tilde{\sigma}_{13}(V_{21} - V_{23}) & 0 & \sigma_{13} \tilde{\sigma}_{13}(V_{23} - V_{21}) & 0 \\
0 & \sigma_{14} \tilde{\sigma}_{14}(V_{22} - V_{24}) & 0 & \sigma_{14} \tilde{\sigma}_{14}(V_{24} - V_{22})
\end{smallmatrix}\right].
\end{equation*}
Each entry is weighted by firm~1's switching probability sensitivity
$\sigma_{1k} \tilde{\sigma}_{1k}$ times the difference in firm~2's
values that appears in row $k$ of $\Sigma_1$.

For the block $DT_{22} = \frac{\partial T_2}{\partial V_2}$, from $T_2 = \frac{1}{\rho + \eta}[u_2 + \lambda_1 \Sigma_1(V_1) V_2 + \lambda_2 E_2(V_2)]$, by symmetry with $DT_{11}$:
\begin{equation*}
DT_{22} = \frac{1}{\rho + \eta}\left[ \lambda_1 \Sigma_1(V_1) + \lambda_2 \frac{\partial E_2}{\partial V_2} \right]
\end{equation*}
where the social surplus derivative matrix for firm~2 is:
\begin{equation*}
\frac{\partial E_2}{\partial V_2} = \begin{bmatrix}
\tilde{\sigma}_{21} & \sigma_{21} & 0 & 0 \\
\sigma_{22} & \tilde{\sigma}_{22} & 0 & 0 \\
0 & 0 & \tilde{\sigma}_{23} & \sigma_{23} \\
0 & 0 & \sigma_{24} & \tilde{\sigma}_{24}
\end{bmatrix}.
\end{equation*}
Combining the two terms, the complete expression for $DT_{22}$ is:
\begin{equation*}
DT_{22} = \frac{1}{\rho + \eta}\left[\begin{smallmatrix}
\lambda_1 \tilde{\sigma}_{11} + \lambda_2 \tilde{\sigma}_{21} & \lambda_2 \sigma_{21} & \lambda_1 \sigma_{11} & 0 \\
\lambda_2 \sigma_{22} & \lambda_1 \tilde{\sigma}_{12} + \lambda_2 \tilde{\sigma}_{22} & 0 & \lambda_1 \sigma_{12} \\
\lambda_1 \sigma_{13} & 0 & \lambda_1 \tilde{\sigma}_{13} + \lambda_2 \tilde{\sigma}_{23} & \lambda_2 \sigma_{23} \\
0 & \lambda_1 \sigma_{14} & \lambda_2 \sigma_{24} & \lambda_1 \tilde{\sigma}_{14} + \lambda_2 \tilde{\sigma}_{24}
\end{smallmatrix}\right].
\end{equation*}
Since $\partial E_2/\partial V_2 = \Sigma_2(V_2)$ as well,
comparing with~\eqref{eq:Sigma_2}, the two diagonal blocks coincide:
$DT_{22} = \frac{1}{\rho + \eta}[\lambda_1 \Sigma_1(V_1)
+ \lambda_2 \Sigma_2(V_2)] = DT_{11}$.

\bibliographystyle{chicago}
\bibliography{ctunique}

\newcommand{\noop}[1]{}
\begin{thebibliography}{}

\bibitem[\protect\citeauthoryear{Abbring and Campbell}{Abbring and
  Campbell}{2010}]{abbring-campbell-2010}
Abbring, J.~H. and J.~R. Campbell (2010).
\newblock Last-in first-out oligopoly dynamics.
\newblock {\em Econometrica\/}~{\em 78}, 1491--1527.

\bibitem[\protect\citeauthoryear{Abbring, Campbell, Tilly, and Yang}{Abbring
  et~al.}{2018}]{abbring-campbell-tilly-yang-2018}
Abbring, J.~H., J.~R. Campbell, J.~Tilly, and N.~Yang (2018).
\newblock Very simple {Markov}-perfect industry dynamics: Theory.
\newblock {\em Econometrica\/}~{\em 86}, 721--735.

\bibitem[\protect\citeauthoryear{Aguirregabiria}{Aguirregabiria}{2012}]{aguirregabiria-2012}
Aguirregabiria, V. (2012).
\newblock A method for implementing counterfactual experiments in models with
  multiple equilibria.
\newblock {\em Economics Letters\/}~{\em 114}, 190--194.

\bibitem[\protect\citeauthoryear{Aguirregabiria, Collard-Wexler, and
  Ryan}{Aguirregabiria et~al.}{2021}]{aguirregabiria-collardwexler-ryan-2021}
Aguirregabiria, V., A.~Collard-Wexler, and S.~P. Ryan (2021).
\newblock Dynamic games in empirical industrial organization.
\newblock In K.~Ho, A.~Horta\c{c}su, and A.~Lizzeri (Eds.), {\em Handbook of
  Industrial Organization}, Volume~4, Chapter~4, pp.\  225--343. Elsevier.

\bibitem[\protect\citeauthoryear{Aguirregabiria and Marcoux}{Aguirregabiria and
  Marcoux}{2021}]{aguirregabiria-marcoux-2021}
Aguirregabiria, V. and M.~Marcoux (2021).
\newblock Imposing equilibrium restrictions in the estimation of dynamic
  discrete games.
\newblock {\em Quantitative Economics\/}~{\em 12\/}(4), 1223--1271.

\bibitem[\protect\citeauthoryear{Aguirregabiria and Mira}{Aguirregabiria and
  Mira}{2002}]{aguirregabiria-mira-2002}
Aguirregabiria, V. and P.~Mira (2002).
\newblock Swapping the nested fixed point algorithm: A class of estimators for
  discrete {Markov} decision models.
\newblock {\em Econometrica\/}~{\em 70}, 1519--1543.

\bibitem[\protect\citeauthoryear{Aguirregabiria and Mira}{Aguirregabiria and
  Mira}{2007}]{aguirregabiria-mira-2007}
Aguirregabiria, V. and P.~Mira (2007).
\newblock Sequential estimation of dynamic discrete games.
\newblock {\em Econometrica\/}~{\em 75}, 1--53.

\bibitem[\protect\citeauthoryear{Aguirregabiria and Suzuki}{Aguirregabiria and
  Suzuki}{2014}]{aguirregabiria-suzuki-2014}
Aguirregabiria, V. and J.~Suzuki (2014).
\newblock Identification and counterfactuals in dynamic models of market entry
  and exit.
\newblock {\em Quantitative Marketing and Economics\/}~{\em 12}, 267--304.

\bibitem[\protect\citeauthoryear{Aradillas-L{\'o}pez}{Aradillas-L{\'o}pez}{2020}]{aradillas-lopez-2020}
Aradillas-L{\'o}pez, A. (2020).
\newblock The econometrics of static games.
\newblock {\em Annual Review of Economics\/}~{\em 12\/}(1), 135--165.

\bibitem[\protect\citeauthoryear{Arcidiacono, Bayer, Blevins, and
  Ellickson}{Arcidiacono et~al.}{2016}]{abbe-2016}
Arcidiacono, P., P.~Bayer, J.~R. Blevins, and P.~B. Ellickson (2016).
\newblock Estimation of dynamic discrete choice models in continuous time with
  an application to retail competition.
\newblock {\em The Review of Economic Studies\/}~{\em 83}, 889--931.

\bibitem[\protect\citeauthoryear{Bajari, Benkard, and Levin}{Bajari
  et~al.}{2007}]{bajari-benkard-levin-2007}
Bajari, P., C.~L. Benkard, and J.~Levin (2007).
\newblock Estimating dynamic models of imperfect competition.
\newblock {\em Econometrica\/}~{\em 75}, 1331--1370.

\bibitem[\protect\citeauthoryear{Berry}{Berry}{1992}]{berry-1992}
Berry, S.~T. (1992).
\newblock Estimation of a model of entry in the airline industry.
\newblock {\em Econometrica\/}~{\em 60}, 889--917.

\bibitem[\protect\citeauthoryear{Besanko, Doraszelski, Kryukov, and
  Satterthwaite}{Besanko et~al.}{2010}]{besanko10learning}
Besanko, D., U.~Doraszelski, Y.~Kryukov, and M.~Satterthwaite (2010).
\newblock Learning-by-doing, organizational forgetting, and industry dynamics.
\newblock {\em Econometrica\/}~{\em 78}, 453--508.

\bibitem[\protect\citeauthoryear{Blevins}{Blevins}{2026}]{ctgames}
Blevins, J.~R. (2026).
\newblock Identification and estimation of continuous-time dynamic discrete
  choice games.
\newblock {\em Quantitative Economics\/}~{\em 17}, 254--296.

\bibitem[\protect\citeauthoryear{Blevins and Kim}{Blevins and
  Kim}{2024}]{blevins-kim-2024}
Blevins, J.~R. and M.~Kim (2024).
\newblock Nested pseudo likelihood estimation of continuous-time dynamic
  discrete games.
\newblock {\em Journal of Econometrics\/}~{\em 238}, 105576.

\bibitem[\protect\citeauthoryear{Borkovsky, Doraszelski, and Kryukov}{Borkovsky
  et~al.}{2010}]{borkovsky10homotopy}
Borkovsky, R.~N., U.~Doraszelski, and Y.~Kryukov (2010).
\newblock A user's guide to solving dynamic stochastic games using the homotopy
  method.
\newblock {\em Operations Research\/}~{\em 58}, 1116--1132.

\bibitem[\protect\citeauthoryear{Borkovsky, Doraszelski, and Kryukov}{Borkovsky
  et~al.}{2012}]{borkovsky-doraszelski-kryukov-2012}
Borkovsky, R.~N., U.~Doraszelski, and Y.~Kryukov (2012).
\newblock A dynamic quality ladder model with entry and exit: Exploring the
  equilibrium correspondence using the homotopy method.
\newblock {\em Quantitative Marketing and Economics\/}~{\em 10}, 197--229.

\bibitem[\protect\citeauthoryear{Bresnahan and Reiss}{Bresnahan and
  Reiss}{1991}]{bresnahan91empirical}
Bresnahan, T.~F. and P.~C. Reiss (1991).
\newblock Empirical models of discrete games.
\newblock {\em Journal of Econometrics\/}~{\em 48}, 57--81.

\bibitem[\protect\citeauthoryear{Bugni and Bunting}{Bugni and
  Bunting}{2021}]{bugni-bunting-2021}
Bugni, F.~A. and J.~Bunting (2021).
\newblock On the iterated estimation of dynamic discrete choice games.
\newblock {\em The Review of Economic Studies\/}~{\em 88}, 1031--1073.

\bibitem[\protect\citeauthoryear{Cabral and Riordan}{Cabral and
  Riordan}{1994}]{cabral-riordan-1994}
Cabral, L. M.~B. and M.~H. Riordan (1994).
\newblock The learning curve, market dominance, and predatory pricing.
\newblock {\em Econometrica\/}~{\em 62}, 1115--1140.

\bibitem[\protect\citeauthoryear{Ciliberto and Tamer}{Ciliberto and
  Tamer}{2009}]{ciliberto-tamer-2009}
Ciliberto, F. and E.~Tamer (2009).
\newblock Market structure and multiple equilibria in the airline industry.
\newblock {\em Econometrica\/}~{\em 77}, 1791--1828.

\bibitem[\protect\citeauthoryear{Curtat}{Curtat}{1996}]{curtat-1996}
Curtat, L.~O. (1996).
\newblock {Markov} equilibria of stochastic games with complementarities.
\newblock {\em Games and Economic Behavior\/}~{\em 17}, 177--199.

\bibitem[\protect\citeauthoryear{de~Paula}{de~Paula}{2013}]{depaula-2013}
de~Paula, A. (2013).
\newblock Econometric analysis of games with multiple equilibria.
\newblock {\em Annual Review of Economics\/}~{\em 5}, 107--131.

\bibitem[\protect\citeauthoryear{Dearing and Blevins}{Dearing and
  Blevins}{2025}]{dearing-blevins-2025}
Dearing, A. and J.~R. Blevins (2025).
\newblock Efficient and convergent sequential pseudo-likelihood estimation of
  dynamic discrete games.
\newblock {\em The Review of Economic Studies\/}~{\em 92}, 981--1021.

\bibitem[\protect\citeauthoryear{Doraszelski and Escobar}{Doraszelski and
  Escobar}{2010}]{doraszelski-escobar-2010}
Doraszelski, U. and J.~F. Escobar (2010).
\newblock A theory of regular {Markov} perfect equilibria in dynamic stochastic
  games: {Genericity}, stability, and purification.
\newblock {\em Theoretical Economics\/}~{\em 5}, 369--402.

\bibitem[\protect\citeauthoryear{Doraszelski and Judd}{Doraszelski and
  Judd}{2012}]{doraszelski-judd-2012}
Doraszelski, U. and K.~L. Judd (2012).
\newblock Avoiding the curse of dimensionality in dynamic stochastic games.
\newblock {\em Quantitative Economics\/}~{\em 3}, 53--93.

\bibitem[\protect\citeauthoryear{Doraszelski and Satterthwaite}{Doraszelski and
  Satterthwaite}{2010}]{doraszelski-satterthwaite-2010}
Doraszelski, U. and M.~Satterthwaite (2010).
\newblock Computable {Markov}-perfect industry dynamics.
\newblock {\em The {RAND} Journal of Economics\/}~{\em 41}, 215--243.

\bibitem[\protect\citeauthoryear{Durrett}{Durrett}{2019}]{durrett-2019}
Durrett, R. (2019).
\newblock {\em Probability: Theory and Examples\/} (Fifth ed.).
\newblock Cambridge University Press.

\bibitem[\protect\citeauthoryear{Egesdal, Lai, and Su}{Egesdal
  et~al.}{2015}]{egesdal15dynamic}
Egesdal, M., Z.~Lai, and C.-L. Su (2015).
\newblock Estimating dynamic discrete-choice games of incomplete information.
\newblock {\em Quantitative Economics\/}~{\em 6}, 567--597.

\bibitem[\protect\citeauthoryear{Ericson and Pakes}{Ericson and
  Pakes}{1995}]{ericson-pakes-1995}
Ericson, R. and A.~Pakes (1995).
\newblock {Markov}-perfect industry dynamics: A framework for empirical work.
\newblock {\em The Review of Economic Studies\/}~{\em 62}, 53--82.

\bibitem[\protect\citeauthoryear{Esp\'{i}n-S\'{a}nchez, Parra, and
  Wang}{Esp\'{i}n-S\'{a}nchez et~al.}{2023}]{espin-sanchez-parra-wang-2023}
Esp\'{i}n-S\'{a}nchez, J.-A., A.~Parra, and Y.~Wang (2023).
\newblock Equilibrium uniqueness in entry games with private information.
\newblock {\em The {RAND} Journal of Economics\/}~{\em 54}, 512--540.

\bibitem[\protect\citeauthoryear{Haller and Lagunoff}{Haller and
  Lagunoff}{2000}]{haller-lagunoff-2000}
Haller, H. and R.~Lagunoff (2000).
\newblock Genericity and {Markovian} behavior in stochastic games.
\newblock {\em Econometrica\/}~{\em 68}, 1231--1248.

\bibitem[\protect\citeauthoryear{Herings and Peeters}{Herings and
  Peeters}{2004}]{herings-peeters-2004}
Herings, P. J.-J. and R.~J. A.~P. Peeters (2004).
\newblock Stationary equilibria in stochastic games: Structure, selection, and
  computation.
\newblock {\em Journal of Economic Theory\/}~{\em 118}, 32--60.

\bibitem[\protect\citeauthoryear{Horst}{Horst}{2005}]{horst-2005}
Horst, U. (2005).
\newblock Stationary equilibria in discounted stochastic games with weakly
  interacting players.
\newblock {\em Games and Economic Behavior\/}~{\em 51\/}(1), 83--108.

\bibitem[\protect\citeauthoryear{Hotz and Miller}{Hotz and
  Miller}{1993}]{hotz93conditional}
Hotz, V.~J. and R.~A. Miller (1993).
\newblock Conditional choice probabilities and the estimation of dynamic
  models.
\newblock {\em The Review of Economic Studies\/}~{\em 60}, 497--529.

\bibitem[\protect\citeauthoryear{Iskhakov, Rust, and Schjerning}{Iskhakov
  et~al.}{2016}]{iskhakov16recursive}
Iskhakov, F., J.~Rust, and B.~Schjerning (2016).
\newblock Recursive lexicographical search: Finding all {Markov} perfect
  equilibria of finite state directional dynamic games.
\newblock {\em The Review of Economic Studies\/}~{\em 83}, 658--703.

\bibitem[\protect\citeauthoryear{Jia}{Jia}{2008}]{jia-2008}
Jia, P. (2008).
\newblock What happens when {Wal-Mart} comes to town: An empirical analysis of
  the discount retailing industry.
\newblock {\em Econometrica\/}~{\em 76}, 1263--1316.

\bibitem[\protect\citeauthoryear{Jun and Pinkse}{Jun and
  Pinkse}{2020}]{jun-pinkse-2020}
Jun, S.~J. and J.~Pinkse (2020).
\newblock Counterfactual prediction in complete information games: Point
  prediction under partial identification.
\newblock {\em Journal of Econometrics\/}~{\em 216}, 394--429.

\bibitem[\protect\citeauthoryear{Kasahara and Shimotsu}{Kasahara and
  Shimotsu}{2008}]{kasahara-shimotsu-2008}
Kasahara, H. and K.~Shimotsu (2008).
\newblock Pseudo-likelihood estimation and bootstrap inference for structural
  discrete {Markov} decision models.
\newblock {\em Journal of Econometrics\/}~{\em 146\/}(1), 92--106.

\bibitem[\protect\citeauthoryear{Kasahara and Shimotsu}{Kasahara and
  Shimotsu}{2012}]{kasahara-shimotsu-2012}
Kasahara, H. and K.~Shimotsu (2012).
\newblock Sequential estimation of structural models with a fixed point
  constraint.
\newblock {\em Econometrica\/}~{\em 80}, 2303--2319.

\bibitem[\protect\citeauthoryear{Magnac and Thesmar}{Magnac and
  Thesmar}{2002}]{magnac-thesmar-2002}
Magnac, T. and D.~Thesmar (2002).
\newblock Identifying dynamic discrete decision processes.
\newblock {\em Econometrica\/}~{\em 70}, 801--816.

\bibitem[\protect\citeauthoryear{McFadden}{McFadden}{1978}]{mcfadden-1978-gev}
McFadden, D.~L. (1978).
\newblock Modelling the choice of residential location.
\newblock In A.~Karlqvist, L.~Lundqvist, F.~Snickars, and J.~W. Weibull (Eds.),
  {\em Spatial Interaction Theory and Planning Models}, Volume~3 of {\em
  Studies in Regional Science and Urban Economics}, pp.\  75--96.
  Am\-ster\-dam: North-Holland.

\bibitem[\protect\citeauthoryear{McFadden}{McFadden}{1981}]{mcfadden81econometric}
McFadden, D.~L. (1981).
\newblock Econometric models of probabilistic choice.
\newblock In C.~F. Manski and D.~L. McFadden (Eds.), {\em Structural Analysis
  of Discrete Data with Econometric Applications}, pp.\  198--272. Cam\-bridge,
  MA: MIT Press.

\bibitem[\protect\citeauthoryear{McKelvey and McLennan}{McKelvey and
  McLennan}{1996}]{mckelvey-mclennan-1996}
McKelvey, R.~D. and A.~McLennan (1996).
\newblock Computation of equilibria in finite games.
\newblock In H.~M. Amman, D.~A. Kendrick, and J.~Rust (Eds.), {\em Handbook of
  Computational Economics}, Volume~1, pp.\  87--142. Amsterdam: Elsevier.

\bibitem[\protect\citeauthoryear{Milgrom and Weber}{Milgrom and
  Weber}{1985}]{milgrom-weber-1985}
Milgrom, P.~R. and R.~J. Weber (1985).
\newblock Distributional strategies for games with incomplete information.
\newblock {\em Mathematics of Operations Research\/}~{\em 10}, 619--632.

\bibitem[\protect\citeauthoryear{Ortega and Rheinboldt}{Ortega and
  Rheinboldt}{1970}]{ortega-rheinboldt-1970}
Ortega, J. and W.~Rheinboldt (1970).
\newblock {\em Iterative Solution of Nonlinear Equations in Several Variables}.
\newblock Academic Press.

\bibitem[\protect\citeauthoryear{Otsu and Pesendorfer}{Otsu and
  Pesendorfer}{2023}]{otsu-pesendorfer-2023}
Otsu, T. and M.~Pesendorfer (2023).
\newblock Equilibrium multiplicity in dynamic games: Testing and estimation.
\newblock {\em The Econometrics Journal\/}~{\em 26}, C26--C42.

\bibitem[\protect\citeauthoryear{Otsu, Pesendorfer, and Takahashi}{Otsu
  et~al.}{2016}]{otsu-pesendorfer-takahashi-2016}
Otsu, T., M.~Pesendorfer, and Y.~Takahashi (2016).
\newblock Pooling data across markets in dynamic {Markov} games.
\newblock {\em Quantitative Economics\/}~{\em 7}, 523--559.

\bibitem[\protect\citeauthoryear{Pesendorfer and Schmidt-Dengler}{Pesendorfer
  and Schmidt-Dengler}{2008}]{pesendorfer08asymptotic}
Pesendorfer, M. and P.~Schmidt-Dengler (2008).
\newblock Asymptotic least squares estimators for dynamic games.
\newblock {\em The Review of Economic Studies\/}~{\em 75}, 901--928.

\bibitem[\protect\citeauthoryear{Pesendorfer and Schmidt-Dengler}{Pesendorfer
  and Schmidt-Dengler}{2010}]{pesendorfer-schmidt-dengler-2010}
Pesendorfer, M. and P.~Schmidt-Dengler (2010).
\newblock Sequential estimation of dynamic discrete games: A comment.
\newblock {\em Econometrica\/}~{\em 78}, 833--842.

\bibitem[\protect\citeauthoryear{Puterman}{Puterman}{2005}]{puterman05markov}
Puterman, M.~L. (2005).
\newblock {\em {Markov} Decision Processes: Discrete Stochastic Dynamic
  Programming}.
\newblock Wiley-Interscience.

\bibitem[\protect\citeauthoryear{Rust}{Rust}{1987}]{rust87optimal}
Rust, J. (1987).
\newblock Optimal replacement of {GMC} bus engines: An empirical model of
  {Harold Zurcher}.
\newblock {\em Econometrica\/}~{\em 55}, 999--1013.

\bibitem[\protect\citeauthoryear{Rust}{Rust}{1994}]{rust-1994-structural}
Rust, J. (1994).
\newblock Structural estimation of {Markov} decision processes.
\newblock In R.~F. Engle and D.~L. McFadden (Eds.), {\em Handbook of
  Econometrics}, Volume~4, pp.\  3081--3143. Am\-ster\-dam: North-Holland.

\bibitem[\protect\citeauthoryear{Shapley}{Shapley}{1953}]{shapley-1953}
Shapley, L.~S. (1953).
\newblock Stochastic games.
\newblock {\em Proceedings of the National Academy of Sciences\/}~{\em 39},
  1095--1100.

\bibitem[\protect\citeauthoryear{Su}{Su}{2014}]{su-2014}
Su, C.-L. (2014).
\newblock Estimating discrete-choice games of incomplete information: Simple
  static examples.
\newblock {\em Quantitative Marketing and Economics\/}~{\em 12\/}(2), 167--207.

\bibitem[\protect\citeauthoryear{Su and Judd}{Su and Judd}{2012}]{su-judd-2012}
Su, C.-L. and K.~L. Judd (2012).
\newblock Constrained optimization approaches to estimation of structural
  models.
\newblock {\em Econometrica\/}~{\em 80}, 2213--2230.

\bibitem[\protect\citeauthoryear{Sweeting}{Sweeting}{2009}]{sweeting-2009}
Sweeting, A. (2009).
\newblock The strategic timing incentives of commercial radio stations: An
  empirical analysis using multiple equilibria.
\newblock {\em The {RAND} Journal of Economics\/}~{\em 40}, 710--742.

\bibitem[\protect\citeauthoryear{Tamer}{Tamer}{2003}]{tamer-2003}
Tamer, E. (2003).
\newblock Incomplete simultaneous discrete response model with multiple
  equilibria.
\newblock {\em The Review of Economic Studies\/}~{\em 70}, 147--165.

\end{thebibliography}

\end{document}